\documentclass[journal]{IEEEtran}
\ifCLASSINFOpdf
\else
\fi
\normalsize
\usepackage{graphicx}
\usepackage{textcomp}
\usepackage{bbm}
\usepackage{adjustbox}
\usepackage{enumitem}
\usepackage{cite,algorithm, algorithmic, amsmath, amssymb, amsthm, empheq, mhsetup}
\usepackage{subfigure,amsfonts, balance}
\usepackage{epstopdf}
\usepackage{enumitem}
\usepackage{setspace}
\usepackage[dvipsnames]{xcolor}
\usepackage{array}
\usepackage{bm}
\usepackage{multirow}
\usepackage{stfloats}

\DeclareMathOperator{\EEE}{\mathbb{E}}

\DeclareMathOperator{\C}{\mathbb{C}}

\DeclareMathOperator{\aaa}{\mathbf{a}}
\DeclareMathOperator{\FF}{\mathcal{F}}

\DeclareMathOperator{\K}{\mathcal{K}}

\DeclareMathOperator{\I}{\mathbf{I}}
\DeclareMathOperator{\OO}{\mathcal{O}}

\DeclareMathOperator{\vv}{\mathbf{v}}
\DeclareMathOperator{\w}{\mathbf{w}}
\DeclareMathOperator{\z}{\mathbf{z}}

\DeclareMathOperator{\HH}{\mathbf{H}}
\DeclareMathOperator{\HHH}{\mathcal{H}}

\DeclareMathOperator{\RRR}{\mathbb{R}}

\DeclareMathOperator{\G}{\mathbf{G}}

\DeclareMathOperator{\LL}{\mathcal{L}}
\DeclareMathOperator{\SSS}{\mathcal{S}}
\DeclareMathOperator{\WWW}{\mathcal{W}}
\DeclareMathOperator{\CN}{\mathcal{CN}}

\DeclareMathOperator{\B}{\mathbf{B}}

\DeclareMathOperator{\BB}{\mathsf{B}}

\DeclareMathOperator{\MM}{\mathcal{M}}

\DeclareMathOperator{\e}{\mathbf{e}}
\DeclareMathOperator{\DD}{\mathcal{D}}

\DeclareMathOperator{\PPP}{\mathcal{P}}

\DeclareMathOperator{\rr}{\mathbf{r}}
\DeclareMathOperator{\x}{\mathbf{x}}
\DeclareMathOperator{\ttt}{\mathbf{t}}

\DeclareMathOperator{\y}{\mathbf{y}}

\DeclareMathOperator{\uu}{\mathbf{u}}
\DeclareMathOperator{\Q}{\mathbf{Q}}
\DeclareMathOperator{\cc}{\mathsf{C}}

\DeclareMathOperator{\g}{\mathbf{g}}

\DeclareMathOperator{\pphi}{\boldsymbol{\phi}}

\DeclareMathOperator{\THeta}{\boldsymbol{\theta}}

\DeclareMathOperator{\BETA}{\boldsymbol{\beta}}

\DeclareMathOperator{\SE}{\mathtt{SE}}

\newtheorem{proposition}{Proposition}

\ifCLASSINFOpdf
\else
\fi

\def\BibTeX{{\rm B\kern-.05em{\sc i\kern-.025em b}\kern-.08em
   T\kern-.1667em\lower.7ex\hbox{E}\kern-.125emX}}

\allowdisplaybreaks
\begin{document}
\title{Joint User Association and Power Control for Cell-Free Massive MIMO}

\author{Chongzheng~Hao,~\IEEEmembership{Graduate Student Member,~IEEE,}
        Tung Thanh~Vu,~\IEEEmembership{Member,~IEEE,}
        \\Hien Quoc~Ngo,~\IEEEmembership{Senior Member,~IEEE,} Minh N.~Dao,
        Xiaoyu~Dang, Chenghua~Wang,~\IEEEmembership{Member,~IEEE,}
        and~Michail~Matthaiou,~\IEEEmembership{Fellow,~IEEE} 
\thanks{This paper was presented in part at the IEEE European Signal Processing Conference (EUSIPCO) \cite{PerryEusip23}. The work of X. Dang was supported partially by the National Natural Science Foundation of China under Grants 61971221 and 62031017. The work of C. Hao was supported by the China Scholarship Council under Grant 202106830001. The work of H. Q. Ngo was supported by the U.K. Research and Innovation Future Leaders Fellowships under Grant MR/X010635/1. This work is a contribution by Project REASON, a UK Government funded project under the Future Open Networks Research Challenge (FONRC)
sponsored by the Department of Science Innovation and Technology (DSIT).
The work of M. N. Dao was partially supported by the Australian Research Council (ARC), project number DP230101749, by the PHC FASIC program, project number 49763ZL, and by a public grant from the Fondation Math\'ematique Jacques Hadamard (FMJH).
The work of M. Matthaiou was supported by the European Research Council (ERC) under the European Union’s Horizon 2020
research and innovation programme (grant agreement No. 101001331)}
\thanks{Chongzheng Hao, Xiaoyu Dang, and Chenghua Wang are with the College of Electronic and Information Engineering, Nanjing University of Aeronautics and Astronautics, Nanjing, 211106, P. R. China (e-mail: chongzhhao@gmail.com; \{dang, chwang\}@nuaa.edu.cn).}
\thanks{Tung Thanh Vu is with the School of Engineering, Macquarie University, Australia (e-mail: thanhtung.vu@mq.edu.au).}
\thanks{Hien Quoc Ngo and Michail Matthaiou are with the Centre for Wireless Innovation (CWI), Queen’s University Belfast, BT3 9DT, U.K (e-mail: \{hien.ngo, m.matthaiou\}@qub.ac.uk).}
\thanks{Minh N. Dao is with the School of Science, RMIT University, Melbourne, VIC 3000, Australia (e-mail: minh.dao@rmit.edu.au).}
}

\maketitle

\begin{abstract}
This work proposes novel approaches that jointly design user equipment (UE) association and power control (PC) in a downlink user-centric cell-free massive multiple-input multiple-output (CFmMIMO) network, where each UE is only served by a set of access points (APs) for reducing the fronthaul signaling and computational complexity. In order to maximize the sum spectral efficiency (SE) of the UEs, we formulate a mixed-integer nonconvex optimization problem under constraints on the per-AP transmit power, quality-of-service rate requirements, maximum fronthaul signaling load, and maximum number of UEs served by each AP. In order to efficiently solve the formulated problem, we propose two different schemes according to different sizes of the CFmMIMO systems.
For small-scale CFmMIMO systems, we present a successive convex approximation (SCA) method to obtain a stationary solution and also develop a learning-based method (JointCFNet) to reduce the computational complexity.  
For large-scale CFmMIMO systems, we propose a low-complexity suboptimal algorithm using accelerated projected gradient (APG) techniques. Numerical results show that
our JointCFNet can yield similar performance and significantly decrease the run time compared with the SCA algorithm in  small-scale systems. 
The presented APG approach is confirmed to run much faster than the SCA algorithm in the large-scale system while obtaining a SE performance close to that of the SCA approach. Moreover, the median sum SE of the APG method is up to about $2.8$ fold higher than that of the heuristic baseline scheme. 
\end{abstract}


%
\IEEEpeerreviewmaketitle

\section{Introduction}
\IEEEPARstart{C}{ELL-FREE} massive multiple-input multiple-output (CFmMIMO) has been considered a promising solution for future generations of communication systems due to its potential to provide handover-free and uniformly good quality of service (QoS) to all users \cite{matthaiou2021road, zhang2020prospective}. In a CFmMIMO system, a large number of distributed access points (APs) jointly serve a large number of user equipments (UEs) within a given coverage area without cell boundaries in the same time and frequency resources. Each AP connects to the central processing unit (CPU) via a fronthaul link, and all CPUs are connected through a backhaul network \cite{HienTWC2017}. 
Since CFmMIMO avails of high macro-diversity gains, favorable propagation, and channel hardening, it can achieve ubiquitous coverage and substantial spectral and energy efficiencies with simplified resource allocation schemes 
\cite{emil20TWC}.

In the canonical CFmMIMO, all APs are assumed to serve all UEs jointly in the network 
\cite{HienTWC2017, emil20TWC}. However, the fronthaul capacity and computational complexity grow linearly with the number of UEs \cite{emil2020tcom}. A user-centric scheme was introduced in \cite{Buz2017wcl}, 
creating a serving AP set/cluster for each user; specifically, each AP works with different sets of APs serving different UEs while the serving AP set is determined by specific criteria, e.g., system performance \cite{liu2016tvt}, large-scale fading gain \cite{amm2019global}, or a two-stage process \cite{amm2022tcom}. In a nutshell, the distinction between canonical CFmMIMO and user-centric CFmMIMO is that the former defines that each user can be served by all the APs, while the latter assumes that a dynamic AP cluster serves each UE. Since the canonical CFmMIMO is practically unscalable, we concentrate on the user-centric CFmMIMO in this paper. 

Power control (PC) is indispensable for efficient allocation of the radio resources in wireless systems \cite{gje2008twc}. Hence, proper PC can mitigate inter-user interference and improve the system performance in CFmMIMO \cite{zhao2020Access}. Considering the scalability of practical architectures, each AP's fronthaul burden and computational complexity must be acceptable, and hence, the number of UEs served by each AP should be restricted by a smaller number than the total number of UEs $K$ \cite{Demir21, dan2020}. This means that each UE is served by a set of its associated APs instead of all the APs, and this procedure is called user association (UA) or AP selection. Moreover, joint PC and UA were investigated in 
\cite{Hien2018TGCN} to decrease the fronthaul power consumption and improve energy efficiency (EE). Note that the formulated joint problem is generally a mixed integer programming problem \cite{guen2022tcom}, and some iterative-based algorithms have been utilized to find suboptimal solutions in the related literature. However, such algorithms entail high computational complexity, especially in large-scale CFmMIMO systems with large numbers of APs and UEs. In summary, joint PC and UA can improve the performance of CFmMIMO systems, but the available optimization algorithms are time-consuming and might violate the real-time requirements. The main motivation of this paper is to achieve efficient joint PC and UA with high sum spectral efficiency (SE).

\subsection{Review of Related Literature}
\subsubsection{Power Control}
In \cite{HienTWC2017}, the authors studied the max-min PC problem in the downlink of CFmMIMO systems, and a network-wide bisection search-based algorithm was proposed to solve a sequence of convex feasibility problems in each step. The fractional PC policy for the uplink of CFmMIMO networks was developed in \cite{Nik2019icc}. This scheme depends only on large-scale quantities and can be fully distributed and managed. Inspired by weighted minimum mean square error (WMMSE) minimization and fractional programming (FP), a new FP-based algorithm was presented for max-min fair power allocation in \cite{cha2021OpenJ}. Since the computational complexity of the above algorithms grows polynomially with the scale of the system, deep neural network (DNN) based methods were introduced in \cite{Sal2021globalcom, Zaher2023TWC, Lou2022globecom} to solve the optimization problems but with low complexity. In \cite{Sal2021globalcom}, the authors designed a convolutional neural network (CNN) model to approximate the second-order cone program (SOCP) algorithm for the downlink PC with maximum ratio transmission (MRT). A per-AP deployed fully distributed DNN model and a cluster-based centralized DNN model were proposed in \cite{Zaher2023TWC}. The large-scale fading (LSF) coefficients were used as input for the constructed centralized and distributed models, and the model outputs the near-optimal PC coefficients. In \cite{Lou2022globecom}, a graph neural network (GNN) was developed to mimic the SOCP algorithm for downlink max-min PC with MRT beamforming.

\subsubsection{UE Association}
In this space, \cite{buz2020twc} proposed a Frobenius-based user selection strategy by computing the average Frobenius norm between each AP and UE. Each AP selected the served users with larger Frobenius norm and then derived sum-rate and minimum-rate maximization expressions for the uplink and downlink. Moreover, \cite{gar2020Access} used a goodness-of-fit test to design the dynamic AP mode switch strategy and to improve the EE performance for cell-free millimeter wave massive MIMO networks. A dynamic AP turned ON/OFF strategy was developed in  \cite{fem2020access} to maximize the EE for green CFmMIMO networks. 
Furthermore, some DNN-based methods have been studied in the literature \cite{men2021, chang2021vtc, ghi2023tgcn}. The recent work in \cite{men2021} presented a deep reinforcement learning (DRL) scheme to implement AP selection with the reward function based on the quality of service and power consumption. Maximizing the EE by joint cooperation clustering and content caching with perfect instantaneous channel state information (CSI) using the DRL method, was investigated in \cite{chang2021vtc}. In \cite{ghi2023tgcn}, the authors proposed a DNN model to solve the AP selection problem with EE maximization under training error, pilot contamination, and imperfect CSI.

\subsubsection{Joint Power Control and UE Association}
Joint PC and AP selection for the downlink CFmMIMO network were studied in \cite{Hien2018TGCN}. In \cite{hien2018asi}, the authors investigated the joint AP selection and PC optimization problem in Ricean channels to maximize the smallest SE of all UEs and leveraged the bisection algorithm to solve this problem. The proposed scheme can reduce the fronthaul power consumption and significantly improve the total EE with large numbers of APs. In \cite{vu2020icc}, the authors considered the practical EE performance of a CFmMIMO system, and a joint power allocation and AP selection method was developed to minimize the energy consumption subject to some QoS constraints. 

\subsection{Research Gap and Main Contributions}
Although some iterative-based and DNN-based schemes have been proposed in the literature \cite{HienTWC2017, Nik2019icc, cha2021OpenJ, Sal2021globalcom, Zaher2023TWC, Lou2022globecom, buz2020twc, gar2020Access, fem2020access, men2021, chang2021vtc, ghi2023tgcn}, these works mainly focus on either PC or UA. The drawback of \cite{HienTWC2017, Nik2019icc, cha2021OpenJ, Sal2021globalcom, Zaher2023TWC, Lou2022globecom} is that they have only focused on PC problems, while \cite{buz2020twc, gar2020Access, fem2020access, men2021, chang2021vtc, ghi2023tgcn} proposed various methods for UA alone. However, joint PC and UA are essential for deploying CFmMIMO systems. On the other hand, \cite{Hien2018TGCN, hien2018asi, vu2020icc}, explored the joint PC and UA issue to improve the performance for uplink or downlink in CFmMIMO systems. Although these works have considered the joint problem, they assumed perfect fronthaul links, hence, they ignored the limited capacity constraints in practical scenarios.  

Moreover, developing efficient algorithms for joint PC and UE association is a challenging task that needs careful consideration. The iterative-based algorithms for joint PC and UA in \cite{Hien2018TGCN, hien2018asi, vu2020icc}, require substantial computational resources, making them less likely to be implemented in large-scale CFmMIMO networks. Supervised learning-based DNN models are widely used to approximate iterative algorithms and to reduce the computational complexity significantly 
\cite{Sal2021globalcom, Zaher2023TWC, Raj2021ICC}. However, this method is not available for large-scale systems since it requires enormous time resources to generate a sufficient training dataset. For example, we spent around $24$ hours generating $28,000$ training examples using an Intel (R) i$7-9800$X CPU under the parallel mode. It should be mentioned that our designed model (in Section III-B) is simple and requires less training data. In \cite{dan2019camsap}, $2\times10^6$ training examples are needed to train the developed model. Furthermore, unsupervised learning models can obtain solutions without prepared training datasets \cite{men2021, chang2021vtc, ghi2023tgcn}. However, training these models is challenging since the performance largely depends on finding suitable parameters \cite{Mostafa2022wcnc}. In a nutshell, to obtain the solutions for a joint UA and PC problem, the iterative-based algorithms are unscalable for large-scale CFmMIMO systems, while DNN-based models rely on generating massive training examples and designing efficient training algorithms.

In order to fill this gap, we hereafter consider the joint UA and PC in a downlink CFmMIMO system with limited-capacity fronthaul links and local partial protective zero-forcing (PPZF) processing. Note that the local PPZF processing has been verified to offer superior SE compared to alternative processing techniques \cite{Giovanni20TWC}. Then, we derive a sum SE expression and propose two different schemes for the formulated optimization problem according to different scales of the user-centric CFmMIMO network. The main contributions of this work are summarized as follows:
\begin{itemize}
    \item We formulate a mixed-integer nonconvex problem for maximizing the achievable sum SE of the system with local PPZF processing. The problem is subject to a maximum fronthaul signaling load and a maximum number of UEs served at each AP for reducing the computational complexity, minimum rates required to guarantee service quality, and per-AP transmit power constraints. To solve this formulated problem, we propose a successive convex approximation (SCA) algorithm with convergence to the stationary solution.
    \item Considering real-time requirements, we design a DNN-based JointCFNet model to approximate the SCA algorithm in small-scale CFmMIMO systems. Unlike the widely used single variable optimization DNN model, our designed JointCFNet can easily be extended to multiple optimization problems. Also, the JointCFNet can obtain nearly identical performance with the SCA method and can substantially reduce the computational complexity.
    \item We propose an alternative approach to solve the formulated problem using accelerated projected gradient (APG) techniques for large-scale CFmMIMO systems. The APG approach has much lower complexity than the SCA approach, providing acceptable performance.
    \item Numerical results confirm that in small-scale CFmMIMO systems, the JointCFNet can provide high approximation accuracy while the run time reduces by three orders of magnitude compared with the SCA algorithm. The low-complexity APG algorithm offers a SE performance close to that of the SCA algorithm while performing significantly faster than the SCA algorithm in large-scale CFmMIMO scenarios. Moreover, the APG can offer significantly higher SE than the heuristic baseline scheme. 
\end{itemize}

\subsection{Paper Organization and Notation}
The remainder of this paper is organized as follows: Section \ref{Sec:systemmode} introduces our system model and formulates the joint PC and UA problem for maximum sum SE. We derive the SCA algorithm and design the DNN-based JointCFNet for small-scale CFmMIMO systems in Section \ref{Sec:small-scale}. The proposed low-complexity APG solution for large-scale CFmMIMO networks is presented in Section \ref{Sec:large-scale}. Numerical results are shown in Section \ref{Sec:simulation}. Section \ref{Sec:conclusion} concludes the paper.

\textit{Notation:} Matrices and vectors are denoted by bold upper letters and lower case letters, respectively; $\I_{N}$ denotes the $N \times N$ identify matrix, $\mathsf{A} \in \RRR^{a \times b \times c}$ denotes a tensor, while $\mathbb{E}\{\cdot\}$ denotes the statistical expectation. The superscripts $(\cdot)^*$, $(\cdot)^{\text{T}}$ and $(\cdot)^{\text{H}}$ stand for the conjugate, transpose, and conjugate-transpose, respectively. A circular symmetric complex Gaussian matrix having covariance $\Q$ is denoted by $\mathcal{CN}(0,\Q)$; $\mathbb{C}$ and $\mathbb{R}$ denote the complex, and real field, respectively. Table \ref{Tab:Acronyms} and \ref{Tab:Notations} lists the main acronyms and  notations, respectively.

\begin{table}[!t]\small
\renewcommand{\arraystretch}{1.0}
\caption{List of Acronyms
}
\label{Tab:Acronyms}
\begin{center}
\vspace{-4mm}
\begin{tabular}{|c|c|}
\hline
Acronyms & Name \\
\hline
APG & Accelerated projected gradient \\
\hline
APs & Access points \\
\hline
CFmMIMO & Cell-free massive multiple-input multiple-output \\
\hline
CNN & Convolutional neural network \\
\hline
CSI & Channel state information \\
\hline
DL & Deep learning \\
\hline
DNN & Deep neural network \\
\hline
EE & Energy efficiency \\
\hline
FULL & Full user equipment association \\
\hline
FZF & Full-pilot zero-forcing \\
\hline
HEU & Heuristic \\
\hline
LSF & Large-scale fading \\
\hline
MMSE & Minimum mean square error \\
\hline
MRT & Maximum ratio transmission \\
\hline
PC & Power control \\
\hline
PMRT & Protective maximum ratio transmission \\
\hline
PPZF & Partial protective zero-forcing \\
\hline
PZF & Partial zero-forcing \\
\hline
QoS & Quality of service \\
\hline
SCA & Successive convex approximation \\
\hline
SE & Spectral efficiency \\
\hline
TDD & Time-division duplexing \\
\hline
UA & User association \\
\hline
UEs & User equipments \\
\hline
\end{tabular}
\end{center}
\vspace{-2mm}
\end{table}

\begin{table}[!t]\small
\renewcommand{\arraystretch}{1.0}
\caption{List of Main Notations and Definitions
}
\label{Tab:Notations}
\begin{center}
\vspace{-4mm}
\begin{tabular}{|c|c|}
\hline
Notation & Definition \\
\hline
$M$ & Number of APs \\
\hline
$K$ & Number of UEs \\
\hline
{$N$} & Number of antennas at each AP \\
\hline
$\rho_p$ & Normalized pilot power at each AP \\
\hline
$\rho_d$ & Max. normalized transmit power at each AP \\
\hline
${\g}_{mk}$ & Channel vector between AP $m$ and UE $k$ \\
\hline
$\bm{\phi}_{k}$ & Pilot sequence of UE $k$ \\
\hline
{$\tilde {\g}_{mk}$} & Small-scale fading between AP $m$ and UE $k$ \\
\hline
{$\hat {\g}_{mk}$} & {Estimated channel between AP $m$ and UE $k$} \\
\hline
${\beta}_{mk}$ & Large-scale fading between AP $m$ and UE $k$ \\
\hline
$\tau_{p}$ & Length of pilot sequence \\
\hline
$\tau_{c}$ & Length of coherence block \\
\hline
$\SSS_{m}$ & Subset of strong UEs \\
\hline
$\WWW_{m}$ & Subset of weak UEs \\
\hline
$\uu_{mk}^{\text{PZF}}$ & Precoding vector with PZF \\
\hline
$\uu_{ml}^{\text{PMRT}}$ & Precoding vector with PMRT \\
\hline
$\theta_{mk}$ & Power control factor of AP $m$ and UE $k$\\
\hline
{$s_k$, $s_{\ell}$} & {Transmitted data symbol of strong and weak UEs}\\
\hline
$a_{mk}$ & Association coefficient of UE $k$ and AP $m$ \\
\hline
$\SE_k$ & Downlink SE for UE $k$ \\
\hline
$\SE_{\text{QoS}}$ & Threshold SE of each UE \\
\hline
$C_{\max}$ & Max. fronthaul load of each AP\\
\hline
$\widehat{K}_m$ & Max. number of UEs served by each AP $m$ \\
\hline
\end{tabular}
\end{center}
\vspace{-2mm}
\end{table}

\section{System Model And Problem Formulation} \label{Sec:systemmode}
We consider a CFmMIMO system, where $M$ APs simultaneously serve $K$ single-antenna UEs in the same frequency band using time-division duplexing (TDD) \cite{HienTWC2017}. Each AP is equipped with $N$ antennas and randomly distributed in a large geographical area, connected to the CPUs via the fronthaul links (see Fig~\ref{fig:Systemmodel}). We assume a block-fading channel model where the channel response can be approximated as constant in a coherence interval and varies independently between different intervals \cite{EmilMIMO}. In TDD operation mode, each coherence block includes three main phases: uplink training, uplink payload data, and downlink payload data transmission. Here, we focus on the downlink, and hence, uplink payload data transmission is neglected.
In the uplink training phase, the UEs send their pilot signals to the APs, and each AP estimates the channels to the corresponding users. The acquired channel estimates are used to precode the transmit signals in the downlink.
We assume that the fronthaul links can offer error-free transmission, yet the fronthaul capacity is limited. Moreover, the system scale is related to the number of APs $M$ and UEs $K$ in the CFmMIMO systems \cite{farooq21TCOM}. In this work, we adopt the definition of system scale as $M \times K$ from \cite{Mai2022TWC}. Referring to \cite{Sal2021globalcom, Raj2021ICC} and \cite{Mai2022TWC}, we assume that 
$MK<1000$ in a small-scale CFmMIMO system, while for a large-scale CFmMIMO system, 
$MK \geq 1000$.
\begin{figure}[t]
	\centering
	\vspace{0em}
    \includegraphics[width=80mm]{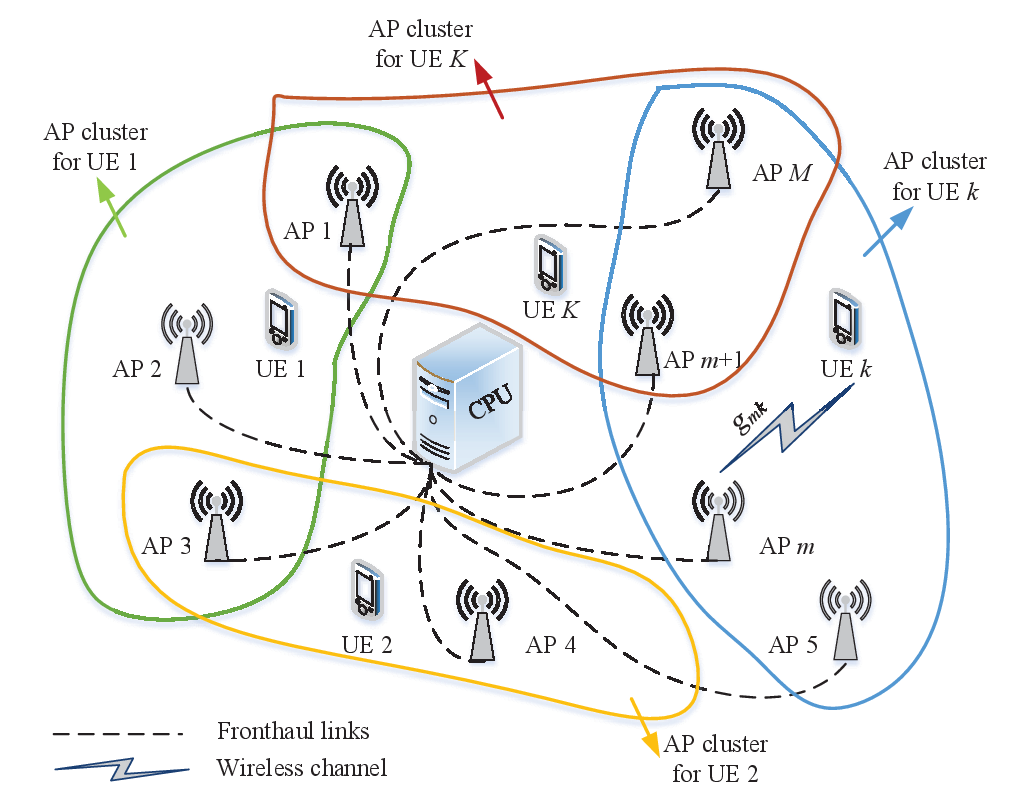}
	\vspace{0em}
	\caption{Architecture of the CFmMIMO network. The distributed APs connect to the CPU via fronthaul links, while the APs may jointly serve UEs.}
	\label{fig:Systemmodel}
\end{figure}

\subsection{Uplink Channel Estimation}
In each coherence block of length $\tau_c$ symbols, all the UEs simultaneously send their pilots of length $\tau_p$ symbols to the APs. We assume that the pilots are pairwisely orthogonal, which requires $\tau_p \geq K$, and ${{\pphi_k^{\text{H}}{\pphi}_k}} = \tau _p$, ${{\pphi_k^{\text{H}}{\pphi}_{\ell}}} = 0, \forall \ell \neq k$. Let $\sqrt {{\rho _p}} \bm{{\phi}} _k \in {\mathbb{C}^{{\tau _p} \times 1}}$ be the pilot sequence sent by the $k$-th user, where ${\rho _p}$ is the normalized pilot power at each user and we assume that the UEs transmit pilot with full power. The received pilot at the $m$-th AP is
\begin{align}
    \label{receivedPilot}
    {\mathbf{Y}}_m^p = \sum\limits_{k = 1}^K {\g}_{mk} \sqrt {{\rho _p}} \bm{\phi} _k^{\text{H}} + {\mathbf{N}}_m^p,
\end{align}
where ${\mathbf{N}}_m^p \in {\mathbb{C}^{N \times {\tau _p}}}$ is the circularly symmetric complex white Gaussian noise matrix, whose elements are independent and identically distributed $\mathcal{C}\mathcal{N}\left( {0,1} \right)$; ${\g}_{mk} \in {\mathbb{C}^{N \times 1}}$ is the channel between the $m$-th AP and the $k$-th UE modeled as
\begin{align}
    \label{ChannelModel}
    {\g}_{mk} = {\left( {{\beta _{mk}}} \right)^{1/2}}\tilde {\g}_{mk},
\end{align}
where ${\beta_{mk}}$ represents the LSF coefficients between AP $m$ and UE $k$, with $m \in \mathcal{M} \triangleq \left\{ {1, \ldots ,M} \right\}$, and $k \in \mathcal{K} \triangleq \left\{ {1, \ldots K} \right\}$, and  $\tilde {\g}_{mk} \in {\mathbb{C}^{N \times 1}}$ denotes the small-scale fading, distributed as $\tilde {\g}_{mk} \sim \mathcal{C}\mathcal{N}\left( {{\mathbf{0}},{{\mathbf{I}}_N}} \right)$.

At AP $m$, ${\g}_{mk}$ is estimated using the received pilot signals and the minimum mean-square error (MMSE) estimation technique. By following \cite[Corollary B.18]{EmilMIMO}, \cite{kay97}, the projection of ${\mathbf{Y}}_m^p$ onto ${\bm{\phi}_k}$ is executed, then using the MMSE yields
\begin{align}
    \label{MMSEEstimation}
    \hat {\g}_{mk} = {c_{mk}}{\mathbf{Y}}_m^p {\pphi} _k,
\end{align}
where ${c_{mk}}$ is defined as
\begin{align}
    \label{MMSECoefficent}
    {c_{mk}} \triangleq \frac{{\sqrt {{\rho _p}} {\beta _{mk}}}}{{{\tau _p}{\rho _p}{\beta _{mk}} + 1}}.
\end{align}
Let $\mathbf{y}_{mk} = {\mathbf{Y}}_m^p {\pphi} _k = \sum\limits_{k = 1}^K {\mathbf{g}}_{mk}\sqrt{{\rho_p}} \tau _p + {{\mathbf{n}}_m^p}'$, where $\mathbf{y}_{mk}\in {\mathbb{C}^{N \times {1}}}$ and ${{\mathbf{n}}_m^p}' = {\mathbf{N}}_m^p\pphi_k \sim \mathcal{C}\mathcal{N}\left( {{\mathbf{0}},{{\tau _p\mathbf{I}}_N}} \right)$. By assuming an independent Rayleigh fading channel and that the noise elements are statistically independent,
we have that ${[\mathbf{y}_{mk}]_{n}} \sim \mathcal{C}\mathcal{N}\left(0, \tau_p\left(\tau_p\rho_p\beta_{mk} + 1 \right)\right)$.
Thus, the estimated channel $\hat {\g}_{mk}$ is distributed according to $\mathcal{C}\mathcal{N}\left( {{\mathbf{0}},\sigma _{mk}^2{{\mathbf{I}}_N}} \right)$, where $\sigma_{mk}^2$ is the mean-square of the estimate corresponds to any $n$-th antenna element is denoted by
\begin{align}
    \label{EstimateMeaNSquare}
    \begin{split}
    \sigma _{mk}^2 & \triangleq \mathbb{E}\left\{ {{{\left| \left[ \hat {\g}_{mk} \right]_n \right|}^2}} \right\} = {c_{mk}^{2} \mathbb{E}\left\{{{{\left| \left[\mathbf{y}_{mk} \right]_n \right|}^2}} \right\}} \\ & = \frac{{{\tau _p}{\rho _p}\beta _{mk}^2}}{{{\tau _p}{\rho _p}\beta _{mk} + 1}}.
    \end{split}
\end{align}
After that, let $\hat{\G}_m \triangleq [\hat{\g}_{m1},\dots,\hat{\g}_{mK}]$ be the estimate of the channel matrix between all the UEs and AP $m$.

\subsection{Downlink Data Transmission with PPZF Precoding}
After acquiring the channels through the uplink pilot, each AP treats the estimated channel as the actual channel. In the remaining $(\tau_{c} - \tau_{p})$ symbols of each coherence block, the APs first implement precoding and then transmit signals to UEs. In the following part, we use the local PPZF precoding because \textit{(i)} it provides higher SE than the other precoding methods, such as MRT, local full-pilot zero-forcing (FZF), and local partial zero-forcing (PZF) \cite{Giovanni20TWC}; \textit{(ii)} it can be implemented in a distributed fashion with local CSI at the APs. More specifically, in this work, the downlink transmission includes two steps: (S1) precoding design at each AP for all the UEs and (S2) joint UE association and PC for optimizing the SE of the system.

\subsubsection{Step (S1)}
Hereafter, we denote by \textit{strong} UEs the UEs that have the largest channel gains, while we denote by \textit{weak} UEs the UEs that have the smallest channel gains. The key idea of the PPZF technique is that each AP suppresses only the interference of the \textit{strongest} UEs while tolerating the interference of the \textit{weakest} UEs. To this end, AP $m$ with PPZF first categorizes the subsets of strong and weak UEs as $\SSS_m\subset\K$ and $\WWW_m\subset\K$, respectively. Here, $\SSS_m \cap \WWW_m = \varnothing$ and $|\SSS_m| + |\WWW_m| = K$, where $|\SSS|$ is the cardinality of set $\SSS$. The strategy for choosing these subsets is discussed later in Section~\ref{Sec:simulation}. Then, each AP $m$ uses the local PZF technique to precode signals for UEs in $\SSS_m$, and uses a protective MRT (PMRT) technique to serve the UEs in $\WWW_m$. The signal transmitted by AP $m$ is given by 
\begin{align}
    \label{APTransmittedSignal}
    \!\x_m^{\text{PPZF}} \!=\! \sum_{k\in\SSS_m}\! \sqrt{\rho_d } \theta_{mk} \uu_{mk}^{\text{PZF}} s_k + \sum_{\ell\in\WWW_m}\! \sqrt{\rho_d } \theta_{m\ell} \uu_{m{\ell}}^{\text{PMRT}} s_{\ell},
\end{align}
where ${s_i},i \in \left\{ {k,\ell} \right\}$ is the data symbol, which satisfies $\EEE \{ \left| {{s_i}} \right|^2 \} = 1$, $\uu_{mk}^{\text{PZF}} \in \C^{N\times 1}$ and $\uu_{mk}^{\text{PMRT}} \in \C^{N\times 1}$ are the precoding vectors with $\EEE\{\|\uu_{mk}^{\text{PZF}}\|^2\} = \EEE\{\|\uu_{mk}^{\text{PMRT}}\|^2\} = 1$, while $\rho_d$ is the maximum normalized transmit power at each AP and 
\begin{align}
    \label{thetamk}
    \theta_{mk} \geq 0, \forall m,k,
\end{align}
are PC coefficients. The transmitted power at AP $m$ is constrained by $\EEE\{|\x_{m}^{\text{PPZF}}|^2\}\leq \rho_d$ which is equivalent to
\begin{align}
    \label{sumtheta}
    \sum_{k\in\K}\theta_{mk}^2 \leq 1, \forall m.
\end{align}

Let $\hat{\G}_{\SSS_m} \in \C^{N \times |\SSS_m|}$ be the matrix formed by stacking the estimated channels of all UEs in $\SSS_m$ of AP $m$. Denote by $j_{mk}\in\{1, \dots, |\SSS_m|\}$ the index of UE $k$ in $\SSS_m$. Then, we have $\hat{\G}_{\SSS_m}\e_{j_{mk}}=\hat{\g}_{mk}$, where $\e_{j_{mk}}$ is the $j_{mk}$-th column of $\I_{|\SSS_m|}$. We note that the set $\SSS_m$ is independent of the UA. The precoding vector $\uu_{mk}^{\text{PZF}}$ can be expressed as 
\begin{align}
    \label{PZFvector}
    \uu_{mk}^{\text{PZF}} = \sqrt{\sigma_{mk}^2 \Big( N-|\SSS_m| \Big)} \hat{\G}_{\SSS_m} \Big( \hat{\G}_{\SSS_m}^\text{H} \hat{\G}_{\SSS_m} \Big)^{-1} \e_{j_{mk}},
\end{align}
where $N > |\SSS_m|$. 

To fully protect the strongest UEs in $\SSS_m$ from the interference from the weakest UEs in $\WWW_m$, the PPZF technique forces the MRT precoding of the signals of UEs in $\WWW_m$ to take place in the orthogonal complement of $\hat{\G}_{\SSS_m}$. Let 
\begin{align}
    \label{PPZFProjetcionMatrix}
    \B_m = \I_M - \hat{\G}_{\SSS_m} \Big( \hat{\G}_{\SSS_m}^\text{H} \hat{\G}_{\SSS_m} \Big) ^{-1} \hat{\G}_{\SSS_m}^\text{H},
\end{align}
be the projection matrix onto the orthogonal complement of $\hat{\G}_{\SSS_m}$. Then, the PMRT  precoding vector for UE $\ell$ at AP $m$ is given by 
\begin{align}
    \label{PMRTvector}
    \uu_{m{\ell}}^{\text{PMRT}} = 
    \frac{\B_m \bar{\HH}_m \e_{\ell}}{\sqrt{\sigma_{m\ell}^2 \Big( N-|\SSS_m| \Big) }},
\end{align}
where $\bar {\HH}_m$ is the full-rank channel estimates matrix of $m$-th AP, $\e_{\ell}$ is the $\ell$-th column of $\I_M$, and $\hat{\g}_{mk}^\text{H} \B_m = \mathbf{0}^\text{T}$ if $k\in \SSS_m$.

At UE $k$, the received signal is
\begin{align} \label{UE-ReceivedSignal}
\begin{split}
    y_k =  \sqrt{\rho_d } & \sum_{m\in\MM} \sum_{k\in\SSS_m}\!  \theta_{mk} \g_{mk}^\text{H} \uu_{mk}^{\text{PZF}} s_k \\ 
           & + \sqrt{\rho_d } \sum_{m\in\MM} \sum_{\ell\in\WWW_m}\!  \theta_{m\ell} \g_{mk}^\text{H}
        \uu_{m{\ell}}^{\text{PMRT}}s_{\ell} + n_k,
\end{split}    
\end{align}
where $n_k \sim \CN(0,1)$ is the additive Gaussian noise. 
Following \cite{Giovanni20TWC}, the downlink achievable SE for UE $k$ is expressed as
\begin{equation}\label{DownlinkSE}
    \SE_k \left( \THeta  \right) = \frac{{{\tau _c} - {\tau _p}}}{{{\tau _c}}}{\log _2} \Big( {1 + \mathrm {SINR_k} } \Big),
\end{equation}
where
\begin{align}
    \label{DownlinkSINR}
    \mathrm {SINR_k} \!=\! \frac{{{{\left( {{U_k}\left( \THeta  \right)} \right)}^2}}}{{{V_k}\left( \THeta  \right)}} \!=\! \frac{{{{\left( {\sum\limits_{m \in \mathcal{M}} {\sqrt {{\rho _d}\left( {N - \left| {{\mathcal{S}_m}} \right|\sigma _{mk}^2} \right)} {\theta _{mk}}} } \right)}^2}}}{{\sum\limits_{l \in \mathcal{K}} {\sum\limits_{m \in \mathcal{M}} {{\rho _d}\theta _{ml}^2\Big( {{\beta _{mk}} - {\delta _{mk}}\sigma _{mk}^2} \Big) \!+\! 1} } }},
\end{align}
is the effective signal-to-interference-and-noise ratio (SINR), $\THeta \triangleq [\THeta_{1}^\mathrm{T},\dots,\THeta_{M}^\mathrm{T}] ^\mathrm{T}, \THeta_{m} \triangleq \left[ \theta_{m1}, \dots, \theta_{mK} \right] ^\mathrm{T}$, $\delta_{mk} = 1 $ if AP $m$ uses PZF for UE $k\in \SSS_m$, and $\delta_{mk} = 0$ if AP $m$ uses PMRT for UE $k\in \WWW_m$.

\subsubsection{Step (S2)}
Although most works assume that the fronthaul links offer infinite capacity \cite{HienTWC2017}, \cite{Hien2018TGCN}, \cite{cha2021OpenJ}, \cite{Giovanni20TWC}, in practical CFmMIMO systems, the fronthaul capacity is subject to some certain limitations \cite{Bashar2019TCom}. Thus, the fronthaul signal load to each AP should be limited and the number of UEs served by each AP should be restricted by a number smaller than $K$ to reduce the complexity of each AP \cite{Demir21}. This also means that each UE is only served by a set of its associated APs instead of all the APs. To this end, we define the association of UE $k$ and AP $m$ as
\begin{align}
    \label{a}
    a_{mk} \triangleq
    \begin{cases}
      1, & \text{if UE $k$ associates with AP $m$}\\
      0, & \mbox{otherwise}
    \end{cases}, \forall m,k.
\end{align}
We have
\begin{align}
\label{thetaa:relation}
    \Big( \theta_{mk}^2 = 0, \forall k,\,\, \text{if}\,\, a_{mk}= 0 \Big),\quad \forall m,
\end{align}
to guarantee that if AP $m$ does not associate with UE $k$, the transmit power $\rho_d \theta_{mk} ^2$ towards  UE $k$ is zero. 
Note that since Steps (S1) and (S2) are performed separately, the UE association in this step does not affect the precoding signals in Step (S1) and the mathematical structure of $\SE_k$ in \eqref{DownlinkSE}. The UE association variable $a_{mk}$ only affects $\SE_k$ via power $\theta_{mk}^2$ and \eqref{thetaa:relation}. Here, we have
\begin{align}
\label{SEQoS}
    &\SE_k\left( \THeta  \right) \geq \SE_{{\text{QoS}}},\forall k,
    \\
    \label{fronthaul:cons}
    & \sum_{k\in\K} a_{mk} \SE_k (\THeta) \leq C_{\max}, \forall m,
    \\
    \label{Khat}
    & \sum_{k\in\K} a_{mk} \leq \widehat{K}_m, \forall m,
    \\
    \label{atleastoneAP:cons}
    & \sum_{m\in\MM} a_{mk} \geq 1, \forall k,
\end{align}
to guarantee that: the SE of each UE needs to be larger than a threshold ${\text{S}}{{\text{E}}_{{\text{QoS}}}}$, the fronthaul load at each AP $m$ is below a threshold $C_{\max}$, the maximum number of UEs served by each AP $m$ is $\widehat{K}_m$, and that each UE is served by at least one AP.

\subsection{Problem Formulation}
In this subsection, we aim at optimizing UA coefficients $\aaa$, where $\aaa \triangleq \{a_{mk}\}, \forall m,k$, and the PC coefficients $\THeta$ to maximize the sum SE. Specifically, we formulate an optimization problem as follows:
\begin{align}\label{P:SE}
    \underset{\aaa,\THeta}{\max}\,\, &
     \sum_{k\in\K} \SE_k(\THeta)\\
    \mathrm{s.t.} \,\,
    \nonumber
    & \eqref{thetamk}, \eqref{sumtheta}, \eqref{a}-\eqref{atleastoneAP:cons},
\end{align}
where $\mathtt{SE}_k(\THeta)$ is the downlink achievable SE for UE $k$, which has been defined in \eqref{DownlinkSE}.

It should be mentioned that our objective \eqref{P:SE} contains some practical constraints of the CFmMIMO system. We consider the limited fronthaul capacity, the per-user QoS requirement, and only statistical CSI knowledge on the CPUs. The problem (21) is a mixed-integer nonconvex optimization problem due to the nonconvex objective function, nonconvex constraints, and binary variables involved. Moreover, there is a tight coupling between the continuous parameters $\THeta$ and binary variables $\aaa$, which makes problem (21) even more complicated. Hence, finding the globally optimal solution to the problem (21) is very challenging.
To efficiently solve the above joint UA and PC problem, we propose different methods depending on the different sizes of CFmMIMO networks. For small-scale CFmMIMO systems, we first propose an SCA-based algorithm and then a deep learning (DL) scheme to reduce the computational complexity, as well as satisfy the real-time requirements in practical deployments. However, employing the DNN-based method in a large-scale CFmMIMO system is impractical since generating sufficient training data for guaranteeing an acceptable level is significantly time-consuming. Thus, we propose a low-complexity APG method to handle the problem \eqref{P:SE} in the large-scale CFmMIMO scenario.

\section{Solution For Small-scale CFmMIMO Systems} \label{Sec:small-scale}
In this section, we consider a small-scale CFmMIMO network and propose two different solutions. 
To solve the problem \eqref{P:SE}, we first propose an SCA method. The SCA technique is classified as an iterative algorithm that can find a stationary solution but has a high computational complexity, making it unsuitable for most practical systems. Hence, we later develop an alternative learning-based approach to achieve approximately the performance of the SCA algorithm, while entailing a lower computational complexity.

First, to deal with the binary constraint \eqref{a}, we see that\cite{Tung2022joint} 
\begin{align}
    \label{BinaryConstraint}
    x\in\{0,1\}\Leftrightarrow x\in[0,1]\,\&\,x-x^2\leq0,
\end{align}
therefore, \eqref{a} can be replaced by 
\begin{align}
    \label{Q}
    & Q (\aaa) \triangleq \sum_{k\in\K} \sum_{m\in\MM} (a_{mk} - a_{mk}^2) \leq 0,
    \\
    \label{a:cons:1}
    & 0 \leq a_{mk}, \forall m,k,
    \\
    \label{a:cons:2}
    & a_{mk} \leq 1, \forall m,k.
\end{align}
In the light of \eqref{sumtheta}, we replace the constraint \eqref{thetaa:relation} by
\begin{align}
    \label{thetamk:2}
    \theta_{mk}^2 \leq a_{mk}, \forall m,k.
\end{align}
Problem \eqref{P:SE} is now equivalent to
\begin{align}
    \label{P:SE:equiv}
    \underset{\x\in\widetilde{\FF}}{\min}\,\, &
    - \sum_{k\in\K} \SE_k(\THeta),
\end{align}
where $\x \triangleq \{\aaa,\THeta\}$, $\widetilde{\FF}\triangleq \{\eqref{thetamk}, \eqref{sumtheta}, \eqref{SEQoS}, \eqref{fronthaul:cons}-\eqref{atleastoneAP:cons}, \eqref{Q}-\eqref{thetamk:2}\}$ is a feasible set.

\subsection{SCA-Based Algorithm} \label{Sec:SCA}
The concept of SCA involves the iterative approximation of a non-convex optimization problem by decomposing the original problem into a sequence of convex sub-problems \cite{Gio2019TWC}. For the sake of the algorithm development, we transform problem \eqref{P:SE:equiv} into a more tractable form and derive an SCA algorithm to solve it. First, we consider the following problem
\begin{align}
    \label{P:SE:equiv:2}
    \underset{\x\in\FF}{\min}\,\, &
     \LL(\x),
\end{align} 
where $\FF\triangleq \{\eqref{thetamk}, \eqref{sumtheta}, \eqref{SEQoS}, \eqref{fronthaul:cons}-\eqref{atleastoneAP:cons}, \eqref{a:cons:1}-\eqref{thetamk:2}\}$ is a feasible set, $\LL(\x) \triangleq - \sum_{k\in\K} \SE_k(\THeta) + \lambda Q(\aaa)$ is the Lagrangian of \eqref{P:SE:equiv}, $\lambda$ is the Lagrangian multiplier corresponding to constraint \eqref{Q}.

\begin{proposition}
\label{proposition-dual}
The value $Q_{\lambda}$ of $Q$ at the solution of \eqref{P:SE:equiv:2} corresponding to $\lambda$ converges to $0$ as $\lambda \rightarrow +\infty$. Also, problem \eqref{P:SE:equiv} has strong duality, i.e.,
\begin{equation}\label{Strong:Dualitly:hold:1}
\underset{\x\in \widetilde{\FF}}{\min}\,\,
-\sum_{k\in\K} \SE_k(\THeta)
=
\underset{\lambda\geq0}{\sup}\,\,
\underset{\x\in{\mathcal{F}}}{\min}\,\,
\LL (\x).
\end{equation}
Then, \eqref{P:SE:equiv:2} is equivalent to \eqref{P:SE:equiv} at the optimal solution $\lambda^* \geq0$ of the sup-min problem in \eqref{Strong:Dualitly:hold:1}.
\end{proposition}
\begin{proof}
     The proof is similar to \cite[Proposition 1] {vu18TCOM} and omitted due to lack of space.
\end{proof}
Also note that $Q_{\lambda}$ must be zero to obtain the optimal solution to \eqref{P:SE:equiv}. As stated in Proposition 1, the optimal solution to \eqref{P:SE:equiv} can be obtained as $\lambda \rightarrow +\infty$. It is acceptable for $Q_{\lambda}$ to be sufficiently small with a sufficiently large value of $\lambda$ in practice. In the simulation part, $\lambda=100$ is enough to ensure that $Q_{\lambda}/(MK) \leq \varepsilon$ with $\varepsilon = 5\times 10^{-5}$. Note that this approach of selecting $\lambda$ has been widely used in the literature \cite{vu18TCOM} and references therein.

We now introduce the new variables $\widehat{\w} \triangleq \{\widehat{w}_k\}$ as
\begin{align}
    \label{Vhat}
    V_k(\THeta) \geq \widehat{w}_k, \forall k,
\end{align}
and we have
\begin{align}
    \label{SE:V}
    \SE_k(\THeta) \leq \widetilde{\SE}_k (\THeta, \widehat{\w}) \triangleq \frac{\tau_c-\tau_p}{\tau_c} \log_2 \Big( 1 + \frac{(U_k(\THeta))^2}{\widehat{w}_k} \Big).
\end{align}
Then, problem \eqref{P:SE:equiv:2} is equivalent to
\begin{subequations}\label{P:SE:equiv:3}
\begin{align}
    \underset{\hat{\x}}{\min}\,\, &
     \widetilde{\LL}(\x,\ttt)
     \\
    \mathrm{s.t.} \,\,
    \nonumber
    & \eqref{thetamk}, \eqref{sumtheta}, \eqref{Khat},\eqref{atleastoneAP:cons}, \eqref{a:cons:1}-\eqref{thetamk:2}, \eqref{Vhat}
    \\
    \label{SE:lb}
    & t_k \leq \SE_k(\THeta), \forall k
    \\
    \label{QoS:cons:2}
    & t_k \geq \SE_{\text{QoS}}, \forall k
    \\
    \label{fronthaul:cons:2}
    & \sum_{k\in\K} a_{mk} \hat{t}_{k} \leq C_{\max}, \forall m
    \\
    \label{that}
    & \widetilde{\SE}_k(\THeta,\widehat{\w}) \leq \hat{t}_k, \forall k,
\end{align}
\end{subequations}
where $\hat{\x} \triangleq \{\aaa, \THeta, \ttt,\hat{\ttt},\widehat{\w}\}$, $\widetilde{\LL}(\x,\ttt) \triangleq - \sum_{k\in\K} t_k + \lambda Q(\aaa)$, $\ttt \triangleq \{t_k\}, \hat{\ttt} \triangleq \{\hat{t}_k\}, \forall k$, are additional variables. Here, \eqref{QoS:cons:2}-\eqref{that} follow \eqref{SEQoS} and  \eqref{fronthaul:cons}.

Note that the function $f( {x,y} ) = \log(1 + \frac{{(x)}^2}{y})$ has a lower bound \cite[Eq. (40)] {Vu2020Wcom}
\begin{align}
     \label{LowerBound:1}
\begin{split}
    f({x,y}) & \geq \log \Big( 1 + \frac{ (x^{(n)})^2 }{y^{(n)}} \Big) - \frac{(x^{(n)})^2}{(y^{n})} \\
             & + 2\frac{x^{(n)}x}{y^{(n)}} - \frac{\Big( (x^{{(n)}})^2 ((x)^2+y) \Big) }{y^{(n)} \Big( (x^{{(n)}})^2 + y^{(n)} \Big)},
\end{split}
\end{align}
where $ x \in \mathbb{R}, y > 0, y^{(n)} > 0$. Now, ${\SE}_k$ has a concave lower bound $\widehat{\SE}_k$ as
\begin{align} \label{SEHat}
\begin{split}
     \widehat{\SE}_k(\THeta) & \triangleq \frac{\tau_c-\tau_p}{\tau_c \log 2} \Bigg[ 
    \log \Big( 1 + \frac{(U_k^{(n)})^2} {\widehat{w}_k^{(n)}} \Big) - \frac{\Big( U_k^{(n)}\Big)^2}{\widehat{w}_k^{(n)}} \\
    & + 2\frac{U_k^{(n)} U_k(\THeta)}{\widehat{w}_k^{(n)}}  - \frac{\Big( U_k^{(n)} \Big)^2 \Big( (U_k(\THeta) \Big)^2 + {\widehat{w}_k( \THeta))}} {{\widehat{w}_k^{(n)}}\Big(\Big( U_k^{(n)}\Big)^2 + {\widehat{w}_k^{(n)}} \Big)} 
    \Bigg].
\end{split}
\end{align}
Thus, constraints \eqref{SE:lb} can be approximated by the following convex constraint
\begin{align}
    \label{SE:lb:convex}
    t_k \leq \widehat{\SE}_k (\THeta), \forall k. 
\end{align}
To deal with constraint \eqref{that}, we notice that the function $f( {x,y} ) = \log(1 + \frac{{(x)}^2}{y})$ has a upper bound
\begin{align} \label{LowerBound:2}
    \begin{split}
    f({x,y}) &\leq \log \Big( (x^{(n)})^2 + y^{(n)} \Big) + \frac{(x^2+y)}{(x^{(n)})^2+y^{(n)})} \\
             & \qquad \quad \qquad \quad \ -1 - \log(y), 
    \end{split}
\end{align}
where $\forall x \geq 0, y > 1$. Then, $\widetilde{\SE}_k(\THeta,\widehat{\w})$ has a convex upper bound as
\begin{align} \label{SEHat:2}
\begin{split}
    \bar{\SE}_k(\THeta,\widehat{\w}) & \triangleq 
\frac{\tau_c-\tau_p}{\tau_c \log 2} \Bigg[  \log \Big( (U_k^{(n)})^2 + \widehat{w}_k^{(n)} \Big) \\
& \quad \quad + \frac{\Big( (U_k(\THeta))^2+\widehat{w}_k \Big)}{\Big( U_k^{(n)} \Big) ^2+\widehat{w}_k^{(n)}} -1 - \log(\widehat{w}_k) \Bigg]. 
\end{split}
\end{align}
\vspace{-2pt}
Therefore, constraint \eqref{that} can be approximated by the following convex constraint
\begin{align}
    \label{that:convex}
    \bar{\SE}_k(\THeta,\widehat{\w}) \leq \hat{t}_k, \forall k.
\end{align} 

By invoking the first-order Taylor series expansion at the point ${a_{mk}^{(n)}}$, $Q(\aaa)$ has a convex upper bound as
\begin{align}
    \widehat{Q} (\aaa) \triangleq \sum_{k\in\K} \sum_{m\in\MM} \Big( a_{mk} - 2a_{mk}^{(n)} a_{mk} + \Big( a_{mk}^{(n)} \Big)^2 \Big).
\end{align}
Similarly, constraints \eqref{Vhat} and \eqref{fronthaul:cons:2} can be respectively approximated by the following convex constraints
\begin{align}
    \label{Vhat:convex}
    & \sum_{\ell\in\K} \!\sum_{m\in\MM} \!\!\rho_d \Big (2\theta_{m\ell}^{(n)} \theta_{m\ell} \!-\! (\theta_{m\ell}^{(n)})^2 \!\Big) \Big( \!\beta_{mk} \!-\! \delta_{mk} \sigma_{mk}^2 \Big) \!\!+\! 1 \!\geq\! \widehat{w}_k, \forall k,
    \\
    \nonumber
    \label{QoS:cons:2:convex}
    & \sum_{k\in\K} 0.25 \Big[ \Big( a_{mk} + \hat{t}_k \Big) ^2-2 \Big( a_{mk}^{(n)} - \hat{t}_k^{(n)} \Big) \Big( a_{mk} - \hat{t}_k \Big)
    \\
    & \qquad\qquad\qquad + \Big( a_{mk}^{(n)} - \hat{t}_k^{(n)} \Big)^2 \Big] \leq C_{\max}, \forall m.
\end{align}
We refer to Appendix A for more details.


\begin{algorithm}[!t]
\caption{Solving problem \eqref{P:SE:equiv:3} using the SCA approach}
\begin{algorithmic}[1]
\label{alg}
\STATE \textbf{Initialize}: $n\!=\!0, \lambda > 1$, and a random $\hat{\x}^{(0)} \in \widehat{\mathcal{F}}$
\REPEAT
\STATE Update $n=n+1$
\STATE Solve \eqref{P:SE:equiv:4:convex} to obtain its optimal solution $\hat{\x}^*$
\STATE Update $\hat{\x}^{(n)}=\hat{\x}^*$
\UNTIL{convergence}
\end{algorithmic}
\end{algorithm}

At the $(n\!+\!1)$-th iteration, for a given point $\hat{\x}^{(n)}$, we can approximate problem \eqref{P:SE:equiv:3} using the following convex problem 
\begin{align}\label{P:SE:equiv:4:convex}
    \underset{\hat{\x}\in\widehat{\FF}}{\min}\,\, &
    \widehat{\LL}(\hat{\x}),
\end{align}
where $\widehat{\LL}(\hat{\x}) \triangleq \!-\! \sum_{k\in\K} t_k + \lambda \widehat{Q} (\aaa)$, $\widehat{\FF} \triangleq \{\eqref{thetamk}, \eqref{sumtheta}, \eqref{Khat}, \eqref{atleastoneAP:cons}, \\ \eqref{a:cons:1}-\eqref{thetamk:2}, \eqref{QoS:cons:2}, \eqref{SE:lb:convex}, \eqref{that:convex}, \eqref{Vhat:convex}, \eqref{QoS:cons:2:convex}\}$ is a convex feasible set. In Algorithm 1, we present the primary procedures implemented to solve the problem \eqref{P:SE:equiv:3}. Starting from a random point $\hat{\x}\in\widehat{\FF}$, we solve problem \eqref{P:SE:equiv:4:convex} to obtain its optimal $\hat{\x}^*$. Then, this solution is used as the initial point to the subsequent iteration. Algorithm 1 will converge to a stationary point, i.e., a Fritz John solution, of problem \eqref{P:SE:equiv:3} (hence \eqref{P:SE:equiv} or \eqref{P:SE}). The proof of this fact is rather standard and follows from \cite{vu18TCOM}.

Problem \eqref{P:SE:equiv:4:convex} can be transformed to an equivalent problem that involves $A_v\triangleq 2MK + 4K$ real-valued scalar variables, $A_l\triangleq 3MK + M + 3K$ linear constraints, $A_q\triangleq MK + 2M + 2K$ quadratic constraints. Therefore, the algorithm that solves problem \eqref{P:SE:equiv:4:convex} requires a complexity of $\OO(\sqrt{A_l+A_q}(A_v+A_l+A_q)A_v^2)$ \cite{tam16TWC}.
\vspace{-4pt}

\subsection{DL-Based Low Complexity Method} \label{Sec:DNN}
In this subsection, we design a DNN model to solve the non-convex joint optimization problem \eqref{P:SE} but with a very low computational complexity that can meet real-time requirements. The proposed DNN model uses supervised learning to emulate Algorithm 1. It is worth mentioning that supervised learning is easy to train and deploy compared with unsupervised learning, while the latter is sometimes unstable, and the performance largely relies on finding the right parameters \cite{Mostafa2022wcnc}. The main idea of the developed scheme is to learn the unknown mapping function between the LSF coefficients, PC, and UA. Thanks to the universal approximation ability, a DNN model can be implemented to mimic arbitrary mapping functions \cite{AR1993TIT}.
    Let $\sigma(\cdot)$ be a bounded, non-constant continuous function, $\mathbf{I}_{m}$ be the m-dimensional hypercube, and $\Omega \left(\mathbf{I}_{m} \right)$ be the space of continuous function on $\mathbf{I}_{m}$. Given any $f \in \Omega \left( \mathbf{I}_{m}\right) $ and error $\varepsilon > 0$, from \cite [Theorem 2]{LESHNO1993}, there exists $N>0$, $\mathbf{p}\in \RRR^{N}$,  $\mathbf{W} \in \RRR^{N \times m}$, $\mathop {\sup_{{\mathbf{X}} \in {{\mathbf{I}}_m}}} \left| {f({\mathbf{X}}) - \mathcal{F}\left( {\mathbf{X}} \right)} \right| < \varepsilon $, where $\FF(\mathbf{X}) = \mathbf{p}^{\text{T}} \sigma({\mathbf{W}{\mathbf{X}}})$ is the mapping function constructed by a neural network model, where $\mathbf{W}$ denotes the trainable parameters. It means that a single hidden layer model can closely approximate any continuous function arbitrarily.

CNNs use local convolutional kernels to extract spatial features effectively and shared weights to reduce the number of parameters. In contrast, the parameters of a fully connected neural network will rapidly increase with the scale of the network. 
Thus, CNNs have served as a cornerstone for significant breakthroughs in DL. Motivated by \cite{Sal2021globalcom, Trinh2020TWC, Bash2020JESAC}, we design a CNN-based JointCFNet to learn the joint power and UA optimization problem, illustrated in Fig.~\ref{fig:JointCFNetmodel}. The developed model offers a huge reduction of the run time and provides approximate performance compared with that of the adopted SCA iterative algorithm.
\begin{figure*}[t]
	\centering
	\vspace{0em}
    \includegraphics[width=160mm]{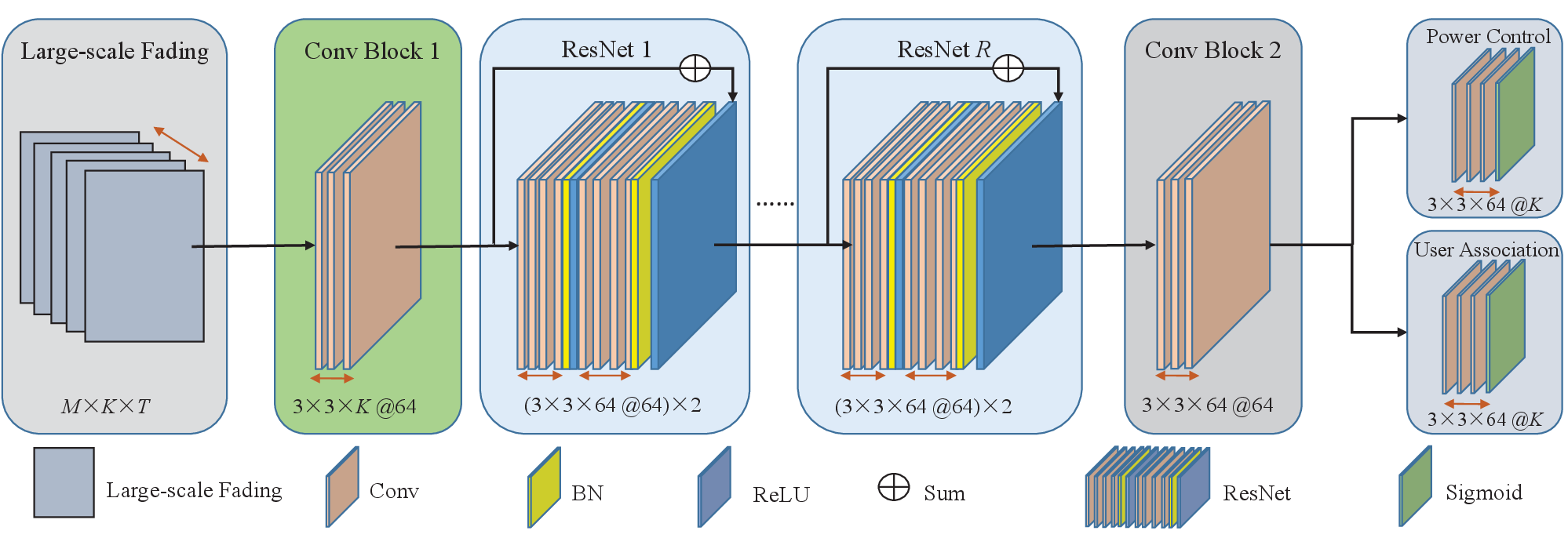}
	\vspace{0em}
	\caption{Structure of the designed JointCFNet for the joint PC and UA.}
	\label{fig:JointCFNetmodel}
\end{figure*}

\subsubsection{Structure of the model} The JointCFNet utilizes the LSF coefficients as input, which is more practical to collect than the instantaneous CSI, and produces the desired PC and UA coefficients as outputs. The convolution blocks are composed of convolutional layers, we define the kernel tensor are $\mathsf{W}_{1} \in \RRR^{3\times{3}\times{K}}$ and $\mathsf{W}_{2} \in \RRR^{3\times{3}\times{64}}$.
Let $@d$ denote the total number of kernels, where their values are $d_1=K$ and $d_2=64$. We omit the bias term from each layer to ensure clarity expression. To guarantee the same dimension between the input and output of each convolutional layer, we set stride and zero padding to 1. Further, we employ the modified residual dense network (ResNet) to fully make use of all the features from the original LSF coefficients. 
Each ResNet comprises convolutional layers, batch normalization layer (BN), and ReLU activation function $f(x) = \max(0,x)$. Finally, the captured information is fed into the PC and UA blocks to calculate the power and UA coefficients using the Sigmoid function. It should be mentioned that designing a DNN topology (e.g., the number of layers, required neurons for each layer, activation function types) to achieve a desired approximation accuracy is also considered as an optimization problem in practice. We vary the number of ResNet, kernel size, filter numbers, and activation functions to find decent configurations with the required mean-square error (MSE) in the training dataset.

    It should be emphasized that the previous works in the literature have proposed various DNN models to perform PC in CFmMIMO. However, these  models are only designed for a single variable optimization, which limits their applications in practice\cite{Sal2021globalcom}, \cite{Lou2022globecom, Bash2020JESAC, Raj2021ICC, Zaher2023TWC}. In this paper, our designed JointCFNet is a two-variable approximator, but it can be easily extended to solve multiple optimization problems. For the multiple variables problems, we can add the desired output blocks after Conv Block 2 (in Fig.~\ref{fig:JointCFNetmodel}), while increasing the convolutional layers and ResNets to capture more channel features.

\subsubsection{Training Phase} To train the JointCFNet, we generate training data and labels using Algorithm 1. In each channel realization, the LSF matrix $\B={[\BETA_{1}^{\text{T}},...\BETA_{M}^{\text{T}}]}^{\text{T}} \in \RRR^{M\times {K}}$, where $\BETA_{m} = {[\beta_{m1},...\beta_{mK}]} \in\RRR^{1 \times{K}}$, the desired PC matrix $\boldsymbol{\Theta}^{\text{opt}}={{[(\THeta_{1}^{\text{opt}})^{\text{T}},...(\THeta_{M}^{\text{opt}})^{\text{T}}}]}^{\text{T}} \in \RRR^{M\times {K}}$, where $\THeta_{m}^{\text{opt}} = {{[\theta_{m1}^{\text{opt}},...\theta_{mK}^{\text{opt}}]}} \in\RRR^{1 \times{K}}$, and the UA matrix $\boldsymbol{\Lambda}^{\text{opt}}={{[(\aaa_{1}^{\text{opt}})^{\text{T}},...(\aaa_{M}^{\text{opt}})^{\text{T}}}]}^{\text{T}} \in \RRR^{M\times {K}}$, where $\aaa_{m}^{\text{opt}} = {{[a_{m1}^{\text{opt}},...a_{mK}^{\text{opt}}}]} \in\RRR^{1 \times{K}}$. One data sample in the training dataset can be written as $\left( \B, \boldsymbol{\Theta}^{\text{opt}}, \boldsymbol{\Lambda}^{\text{opt}}\right)$, where $\B$ is the input, and $\boldsymbol{\Theta}^{\text{opt}}$ and $\boldsymbol{\Lambda}^{\text{opt}}$ are the corresponding output/labels. In the following part, we set $(\mathsf{B}_{j},\mathsf{\Theta}^{\text{opt}}_{j},\mathsf{\Lambda}^{\text{opt}}_{j})$ as the $j$-th training batch randomly selected from the training dataset $\{(\B_{t},\boldsymbol{\Theta}^{\text{opt}}_{t},\boldsymbol{\Lambda}^{\text{opt}}_{t})\}_{t=1}^{T}$,
where $\mathsf{B}_{j}, \mathsf{\Theta}^{\text{opt}}_{j}, \mathsf{\Lambda}^{\text{opt}}_{j}\in \RRR^{M \times K \times b}$, $T$ is the total number of training samples, $b$ is the batch size (number of training samples in this batch), $1 \leq b \leq T$, $j$ is the batch index, $1 \leq j \leq \frac{T}{b}$, and the matrix $\B$ of the $i$-th training sample in $\mathsf{B}_{j}$ can be denoted as $\BB_{j}\left( i \right) \in \RRR^{M \times{K}}$, $1 \leq i \leq b$.  

In the forward propagation, the Conv Block 1 extracts the channel features by using the convolutional layers $\cc_{1} = \text{conv}(\BB_{j}\left( i \right) |\mathsf{W}_{1},d_{2},\Psi_{1})$, where \text{conv} denotes the convolutional operation, $\Psi_{1}$ is the trainable parameters in this layer (in the following, we omit the subscript for convenience), and $\cc_{1}$ is fed to the ResNet1. The input features $\cc_{1}$ first execute convolution under $d_{2}$ kernels $\mathsf{W}_{2}$, giving $\cc_{\text{Res1}}^{\text{Conv1}} = \text{conv}(\cc_{1} |\mathsf{W}_{2},d_{2},\Psi)$. Then, a BN layer normalizes the input data, which benefits the model training. The output of the BN layer can be expressed as $\BB_{\text{Res1}}^{\text{BN1}} = \text{Nor}(\cc_{\text{Res1}}^{\text{Conv1}})$, where $\text{Nor}(\cdot)$ is the normalization operation. Next, the output of the ReLU function, given by $\mathsf{R}_{\text{Res1}}^{\text{ReLU1}}= \max(0, \BB_{\text{Res1}}^{\text{BN1}})$, will be passed to the second convolutional and BN layers. Also note that a summer layer (denoted as $\oplus$ in Fig.~\ref{fig:JointCFNetmodel}) is added before the second ReLU layer, which can be expressed as $\mathsf{S}_{\text{Res1}}^{\text{Sum1}} = \text{sum}(\cc_{1},\BB_{\text{Res1}}^{\text{BN2}})$. To this end, the output of ResNet 1 follows $\cc_{\text{Res}}^{1} = \max(0,\mathsf{S}_{\text{Res1}}^{\text{Sum1}})$. The following $(R-1)$ sequential connected ResNets conduct the same operations, and the final output can be expressed as $\cc_{\text{Res}}^{R} = \max(0,\mathsf{S}_\text{ResR}^{\text{Sum2}})$.

Additionally, another convolutional block is carried out to further extract channel information and can be written as $\cc_{2} = \text{conv}(\cc_{\text{Res}}^{R} |\mathsf{W}_{2},d_{2},\Psi)$. Then, $\cc_{2}$ is finally fed to the PC and UA blocks separately. In each block, the input signal computes convolution and then restricts the output values in the range of $[0,1]$ by using the sigmoid function. Thus, the outputs are $\cc_{\text{power}} = \text{conv}(\cc_{2} |\mathsf{W}_{2},d_{1},\Psi)$, $\cc_{\text{user}} = \text{conv}(\cc_{2} |\mathsf{W}_{2},d_{1},\Psi)$, $\mathsf{Y}_{\text{power}} = \text{sigmoid}(\cc_{\text{power}})$, $\mathsf{Y}_{\text{user}} = \text{sigmoid}(\cc_{\text{user}})$, where the element wise sigmoid function is $\text{sigmoid}(x) = \frac{1}{1+\exp(-x)}$.

In the backward propagation part, we consider the MSE loss over each batch, and the JointCFNet model is trained to minimize the following loss
\begin{align}\label{CNNlossfunction}
    \mathop{\min}\limits_{{\Psi}^{(j)}} J\left( \mathsf{Y}_{\text{power}}, \mathsf{Y}_{\text{user}}|\mathsf{B}_{j}, \Psi^{(j)}\right) = \mathop{\min}\limits_{{\Psi}^{(j)}} J \left(\mathsf{Y}\right),
\end{align}
where
\begin{align}\label{CNNlossfunction:2}
\begin{split}
    \!\!J \left(\mathsf{Y}\right)
    & = \frac{1}{b} \sum\limits_{i=1}^{b}0.5 \left( \mathsf{Y}_{\text{power}}\left(\mathsf{B}_{j}(i)|\Psi^{(j)}\right) \!- \! \Theta^{\text{opt}}(i) \right)^{2} \\ & \quad \quad + \frac{1}{b} \sum\limits_{i=1}^{b}0.5 \left( \mathsf{Y}_{\text{user}}\left(\mathsf{B}_{j}(i)|\Psi^{(j)}\right) \!-\! \Lambda^{\text{opt}}(i) \right)^{2}\!,
    \end{split}
\end{align}
where 
$\Psi^{(j)}$ comprises all the trainable parameters corresponding to $\mathsf{B}_{j}$. The loss function in \eqref{CNNlossfunction} is averaged over all $b$ samples in $\mathsf{B}_{j}$. Here, $\mathsf{Y}_{\text{power}}\left(\mathsf{B}_{j}(i)|\Psi^{(j)})\right) \in \RRR^{M \times{K} \times{1}}$ and $\mathsf{Y}_{\text{user}}\left(\mathsf{B}_{j}(i)|\Psi^{(j)}\right) \in \RRR^{M \times{K} \times{1}}$ are the JointCFNet-based PC and UA coefficients of the $i$-th training sample in $\mathsf{B}_{j}$. Note that $\mathsf{\Theta}^{\text{opt}}(i) \in \RRR^{M \times{K} \times{1}}$ and $\mathsf{\Lambda}^{\text{opt}}(i) \in \RRR^{M \times{K} \times{1}}$ are the corresponding SCA-based PC and UA coefficients. Since PC and UA are equally important in our model, we allocate the same weight to them. We use the Adam optimization \cite{kingma14} to train our data set. More specifically, in the beginning, we initialize all the parameters $\Psi^{(0)}$ using the Gaussian random method. Then, we use stochastic gradient descent 
with momentum $m_{t}$, and learning rate $r_{l}$ to update $\Psi^{(j)}$. 
A brief summary of the training procedure can be found in Algorithm~\ref{algJointCFNet}.

\begin{algorithm}[!t]
\caption{Solving problem \eqref{P:SE} using the JointCFNet}
\begin{algorithmic}[1] \label{algJointCFNet}
\REQUIRE The training dataset $ \DD = \{(\B_{t}, \boldsymbol{\Theta}^{\text{opt}}_{t}, \boldsymbol{\Lambda}^{\text{opt}}_{t})\}_{t=1}^{T}$, batch size $\{b, 1 \leq b \leq T\}$, the number of ResNet blocks $R$, learning rate $r_l$, maximum training epoch $e_{max}$, and the momentum $m_t$  
\STATE \textbf{Initialize}: Initialize the JointCFNet parameters $\Psi^{(0)}$
\REPEAT
\FOR {$j = 1\dots \frac{T}{b}$}
\STATE Generate a training batch $(\mathsf{\mathsf{B}}_{j},\mathsf{\Theta}^{\text{opt}}_{j},\mathsf{\Lambda}^{\text{opt}}_{j})$ 
\STATE Compute ${\mathsf{Y}_\text{power}}\left(\mathsf{B}_{j}(i) | \Psi^{(j)} \right)$ and ${\mathsf{Y}_\text{user}}\left(\mathsf{B}_{j}(i) | \Psi^{(j)} \right)$ in the forward propagation
\STATE Update JointCFNet parameters $\Psi^{(j)}$ using the objective function 
\eqref{CNNlossfunction:2} and the Adam optimization algorithm \cite{kingma14}
in the back propagation
\ENDFOR
\UNTIL{reaching the maximum training epoch $e_{max}$ or the training loss change is negligible in the subsequent epochs}
\end{algorithmic}
\end{algorithm}

\subsubsection{Online Phase and Complexity} Once the training phase ends, the weights and biases are configured as the optimal values. Then, the JointCFNet can compute the PC coefficients $\THeta$ and UA coefficients $\aaa$ under new channel realizations through forward propagation with optimal parameters. This is dissimilar to the conventional iterative algorithms, which must be run from an initial point when the channel changes. The online prediction yields a significant complexity reduction compared with the SCA algorithm and will be compared in Section \ref{Sec:simulation}.

Note that the mathematical structure of the original optimization problem \eqref{P:SE} is complicated with many constraints \eqref{thetamk}, \eqref{sumtheta}, \eqref{a}-\eqref{atleastoneAP:cons}. The output obtained from the trained model might not satisfy some of these constraints. Therefore, we further process the output of the trained model such that the probability of all the constraints being satisfied is as high as possible. Moreover, we observe from our experiments that the constraints that are most likely to be violated are the AP transmit power constraint \eqref{sumtheta} and the AP fronthaul constraint \eqref{fronthaul:cons}. Thus, we process the output by two steps: (i) reduce the transmit power for all the UEs at the APs that have constraint \eqref{sumtheta} violated, and (ii) reduce the UE SEs to have constraint \eqref{fronthaul:cons} satisfied. After processing the output obtained from the trained model by these two steps, we experience in our simulation that there is a high probability of channel realizations that satisfy all the constraints \eqref{thetamk}, \eqref{sumtheta}, \eqref{a}-\eqref{atleastoneAP:cons}. We refer to Section \ref{Sec:simulation} for more detailed discussions. In the following, we explain the above two processing steps. 

Let $\{\theta_{DL,mk}\}$ be the set of PC coefficients obtained from the trained model.
Denote by $\MM_{\eqref{sumtheta}}$ the set of APs that have constraint \eqref{sumtheta} violated. We multiply $\{\theta_{DL,mk}\}_{m\in\MM_{\eqref{sumtheta}}}$ with specific factors to obtain the following PC coefficients  
\begin{align}
    \label{theta:star}
    \theta_{mk}^{\star} = \theta_{DL,mk} \sqrt{\frac{1}{\sum_{k\in\K}\theta_{DL,mk}^2}}, \forall m \in \MM_{\eqref{sumtheta}},
\end{align}
which makes \eqref{sumtheta} satisfied. Now, given the new set of power constraints $\tilde{\mathsf{\Theta}} =\{\theta_{DL,mk}^{\star}\}_{m \in\MM_{\eqref{sumtheta}}} \bigcup \{\theta_{DL,mk}\}_{m \in\MM \setminus \MM_{\eqref{sumtheta}}}$, we compute the new set of the achievable SE of UEs $\{\ddot{\SE}_k\}$ using \eqref{DownlinkSE}. In any channel realization, if $\underset{m\in\MM}{\max} \Big( \sum_{k\in\K} a_{mk} \ddot{\SE}_k \Big) \geq  C_{\max}$, we multiply $\{\ddot{\SE}_k\}$ with factors to obtain the SEs 
\begin{align}
    \label{SEcheck}
    \check{\SE}_k = \ddot{\SE}_k \frac{C_{\max}}{\underset{m\in\MM}{\max} \Big( \sum_{k\in\K} a_{mk} \ddot{\SE}_k \Big)},
\end{align}
while \eqref{SEcheck} leads to 
\begin{align}
    \sum_{k\in\K} a_{mk} \check{\SE}_k \leq C_{\max}, \forall m,
\end{align}
which makes constraint \eqref{fronthaul:cons} satisfied. It should be noted that the step of reducing the UE SEs is reasonable in the space of information theory. It is always feasible to transmit data at the SE that is below the achievable SE obtained from a given set of PC coefficients with an arbitrarily low probability of error.

The computational complexity of the JointCFNet in the online phase is determined by the forward propagation. Let $f_{map}$ be the length of the feature map, $k_{siz}$ denote the kernel size, and $c_{in}$ and $c_{out}$ represent the number of CNN input and output channels, respectively. For a single convolutional layer, the complexity is $\mathcal{O}\left(f_{map}^{2}k_{siz}^{2}c_{in}c_{out}\right)$ \cite{he2015CVPR}. The complexity of our designed JointCFNet is dominated by four modules: the conv blocks, ResNets, PC block, and UA block. It is easy to calculate the complexity of conv blocks1 and block2 as $\mathcal{O}(9MKD_{2})$ and $\mathcal{O}(9MD_{2}^{2})$. Similarly, the complexity of a single ResNet is $\mathcal{O}(18MD_{2}^{2})$. The PC and UA blocks have the same structure; therefore, their complexity is $\mathcal{O}(9MKD_{2})$. Since the convolutional operations mainly determine the complexity, we omit the complexity calculation of the sum and activation functions. According to \cite{Algorithmbook}, we can calculate the rough complexity of JointCFNet as $\mathcal{O}(MKD_{2}^{2})$. We can observe that the complexity of the JointCFNet is negligible compared to the SCA algorithm.

\section{Solution For Large-scale CFmMIMO Systems} \label{Sec:large-scale}
The proposed SCA approach requires solving a series of
convex problems \eqref{P:SE} by interior point methods using off-the-shelf convex solvers. However, these solvers have high complexity and importantly, slow run time when the size of the problem is large (i.e., $MK\geq1000$). 
The DL-based scheme is more time efficient and provides approximation accuracy compared with the iterative algorithms but is a data-hungry model. Specifically, the approximate performance basically relies on the training data size. Collecting sufficient training data and labeling is challenging due to time limitations in practice, especially for a large-scale CFmMIMO network with many APs and UEs. Therefore, in what follows, we propose an alternative approach that has a lower complexity and can find a suboptimal solution to the problem \eqref{P:SE} in a large-scale CFmMIMO system.

\subsection{APG-Based Approach}
We first let $z_{mk}^2 \triangleq a_{mk}, \forall m,k$ and 
\begin{align}
    \label{znonzero}
    0 \leq z_{mk} \leq 1, \forall m,k. 
\end{align}
Constraints \eqref{SEQoS}, \eqref{fronthaul:cons}, \eqref{atleastoneAP:cons}, \eqref{Q} and \eqref{thetamk:2} can be replaced by
\begin{align}
    \label{Q1}
    & Q_1(\z) \triangleq \sum_{k\in\K} \sum_{m\in\MM} (z_{mk}^2 - z_{mk}^4) \leq 0,
    \\
    \label{Q2}
     &Q_2(\THeta) \! \triangleq  \! \sum_{k\in\K}\! \! 
     \Big[\!\max\! \Big(0, \SE_{\text{QoS}} - \SE_k(\THeta)\Big) \Big]^2  \leq 0,
     \\
     \label{Q3}
     \nonumber
     & Q_3 (\THeta,\z) \triangleq \sum_{k\in\K}\! \bigg(\! \Big[\!\max\! \Big(0,1 \!-\!\!\!\! \sum_{m\in\MM}\!\! z_{mk}^2\Big) \Big]^2 \!\!\!   
    \\
    & \qquad \qquad \qquad + \!\!
    \sum_{m\in\MM} \Big[\max(0,\theta_{mk}^2\!-\!z_{mk}^2)\Big]^2 \bigg) \leq 0,
    \\
    \label{Q4}
    & Q_4 (\THeta,\z) \triangleq  \! \sum_{m\in\MM}\! \! 
     \Big[\!\max\! \Big(0, \sum_{k\in\K} z_{mk}^2\SE_{k} (\THeta) - C_{\max} \Big) \Big]^2  \!\leq\! 0.
\end{align}
Therefore, problem \eqref{P:SE:equiv} is equivalent to
\begin{subequations}\label{P:SE:equiv:5}
\begin{align}
    \underset{\vv}{\min}\,\, &
   h(\THeta) \triangleq  - \sum_{k\in\K} \SE_k(\THeta)
    \\
     \mathrm{s.t.} \,\,
     & \eqref{thetamk}, \eqref{sumtheta}, \eqref{znonzero}-\eqref{Q4}
     \\
    \label{Khat:2}
    & \sum_{k\in\K} z_{mk}^2 \leq \widehat{K}_m, \forall m,
\end{align}
\end{subequations}
where \eqref{Khat:2} follows \eqref{Khat}, $\vv \! \triangleq\! [\THeta^\text{T}\!, \z^\text{T}]^\text{T}, \z \triangleq [\z_{1}^\text{T},\dots,\z_{M}^\text{T}]^\text{T},\\ \z_{m} \!\triangleq\! [z_{m1}, \dots, z_{mK}]^\text{T}$. Let $\HHH\triangleq \{\eqref{thetamk}, \eqref{sumtheta}, \eqref{znonzero}-\eqref{Q4}, \eqref{Khat:2}\}$ is the feasible set of \eqref{P:SE:equiv:5}. 

We consider the following problem 
\begin{align}\label{P:SE:equiv:6}
    \underset{\vv\in\widehat{\HHH}}{\min}\,\, &
     f(\vv),
\end{align}
where $\widehat{\HHH}\triangleq \{\eqref{thetamk}, \eqref{sumtheta}, \eqref{znonzero}, \eqref{Khat:2}\}$ is a convex feasible set of \eqref{P:SE:equiv:6}, $f(\vv) \! \triangleq \! h(\THeta) + \! \chi\big[\mu_1 Q_1(\z) \! +  \! \mu_2 Q_2(\THeta) \!+\! \mu_3 Q_3(\THeta,\z) \!+\! \mu_4\\ Q_4(\THeta,\z)\big]$ is the Lagrangian of \eqref{P:SE:equiv:5}, $\mu_1, \mu_2, \mu_3, \mu_4$ are fixed and positive weights,  and $\chi$ is the Lagrangian multiplier corresponding to constraints \eqref{Q1}--\eqref{Q4}.

\begin{proposition}
\label{proposition-dual:2}
The values $Q_{1,\chi},Q_{2,\chi}
,Q_{3,\chi}, Q_{4,\chi}
$ of $Q_1, Q_2, \\Q_3, Q_4
$ at the solution of \eqref{P:SE:equiv:6} corresponding to $\chi$ converge to $0$ as $\chi \rightarrow +\infty$. Also, problem \eqref{P:SE:equiv:5} has strong duality, i.e.,
\begin{equation}\label{Strong:Dualitly:hold}
\underset{\vv\in \HHH}{\min}\,\,
-\sum_{k\in\K} \SE_k(\THeta)
=
\underset{\chi\geq0}{\sup}\,\,
\underset{\vv\in\widehat{\HHH}}{\min}\,\,
f (\vv).
\end{equation}
Then, \eqref{P:SE:equiv:6} is equivalent to \eqref{P:SE:equiv:5} at the optimal solution $\chi^* \geq0$ of the sup-min problem in \eqref{Strong:Dualitly:hold}.
\end{proposition}
\noindent
The proof of Proposition~\ref{proposition-dual:2} follows \cite{vu18TCOM}, and hence, is omitted. Theoretically, $Q_{1, \chi}, Q_{2, \chi}
, Q_{3, \chi}, Q_{4, \chi}$ must be zero for obtaining the optimal solution to \eqref{P:SE:equiv:5}. According to Proposition~\ref{proposition-dual:2}, the optimal solution to \eqref{P:SE:equiv:5} can be obtained as $\chi \rightarrow +\infty$. For practical implementation, it is acceptable for $Q_{1,\chi}/(MK), Q_{2,\chi}/K, Q_{3,\chi}/(MK), Q_{4,\chi}/M \leq \varepsilon$, for some small $\varepsilon$ with a sufficiently large value of $\chi$. In our numerical experiments, initialized $\chi=1$ with $\mu_1 = 50, \mu_2 = 10^3, \mu_3 = 5 \times 10^4, \mu_4 = 10$ is enough to ensure that $Q_{1,\chi}/(MK), Q_{2,\chi}/K, Q_{3,\chi}/(MK), Q_{4,\chi}/M \leq \varepsilon$ with $\varepsilon = 10^{-3}$.

\begin{algorithm}[!t]
\caption{Solving problem \eqref{P:SE:equiv:6} using the APG approach}
\begin{algorithmic}[1]\label{algAPG}
\STATE \textbf{Initialize}: $n\!=\!1$, $q^{(0)} \!=\! 0$, $q^{(1)} \!=\! 1$, random $\vv^{(0)}, \bar{\vv}^{(0)}\!\in\!\widehat{\HHH}$,
$\alpha_{\bar{\vv}} > 0$, $\alpha_{\vv} > 0$, 
$\tilde{\vv}^{(1)} = \vv^{(1)} = \vv^{(0)}$, $\zeta\in[0,1)$, $b^{(1)} = 1$, $\upsilon > 0$, $c^{(1)} = f(\vv^{(1)})$, $\Delta > 1$
\REPEAT
\REPEAT
\STATE Update $\bar{\vv}^{(n)}$ as \eqref{vbar} and  $\tilde{\vv}^{(n+1)}$ as \eqref{vtilde} 
\IF{$f(\tilde{\vv}^{(n+1)}) \leq c^{(n)} - \zeta\| \tilde{\vv}^{(n+1)} - \bar{\vv}^{(n)} \|^2$}
\STATE $\vv^{(n+1)} = \tilde{\vv}^{(n+1)}$
\ELSE
\STATE Update $\hat{\vv}^{(n+1)}$ as \eqref{vhat} 
\STATE Update $\vv^{(n+1)}$ as \eqref{vn+1}
\ENDIF
\STATE Update $q^{(n+1)}$ as \eqref{q}
\STATE Update $b^{(n+1)}$ as \eqref{b} and $c^{(n+1)}$ as \eqref{c}
\STATE Update $n=n+1$
\UNTIL{$\big|\frac{f(\vv^{(n)}) - f(\vv^{(n-10)})}{f(\vv^{(n)})}\big| \!\leq\! \epsilon$ or $\big|\frac{h(\THeta^{(n)}) - h(\THeta^{(n-1)})}{h(\THeta^{(n)})}\big| \!\leq\! \epsilon$}
\STATE Increase $\chi = \chi \times \Delta$
\UNTIL{convergence}
\end{algorithmic}
\vspace{-1mm}
\end{algorithm}

Problem \eqref{P:SE:equiv:6} is ready to be solved by APG techniques. The main steps for solving problem \eqref{P:SE:equiv:6} are outlined in Algorithm~\ref{algAPG}.
Starting with a random point $\vv^{(0)}\!\in\!\widehat{\HHH}$, we compute an extrapolated point for accelerating the convergence of the algorithm as \cite[Eq. (10)]{li15NIPS}
\begin{align}
    \label{vbar}
    \!\!\!\bar{\vv}^{(n)} \!\!=\! \vv^{(n)} \!\!+ \tfrac{q^{(n-1)}}{q^{(n)}} (\tilde{\vv}^{(n)} \!\!-\! \vv^{(n)}) \!\!+\! \tfrac{q^{(n-1)}\!-\!1}{q^{(n)}} (\vv^{(n)} \!\!-\! \vv^{(n-1)}),
\end{align}
where $q^{(n)}$ is an extrapolation parameter in iteration $n$ and computed recursively as
\begin{align}
    \label{q}
    q^{(n+1)} = \frac{1 + \sqrt{4(q^{(n)})^2 + 1}}{2}.
\end{align}
From $\bar{\vv}^{(n)}$, we move along the gradient of the function with a dedicated step size $\alpha_{\bar{\vv}}$. Then, the resulting point $(\bar{\vv} - \alpha_{\bar{\vv}} \nabla f(\bar{\vv}))$ is projected onto the feasible set $\widehat{\HHH}$ to obtain
\begin{align}
    \label{vtilde}
    \tilde{\vv}^{(n+1)} = \PPP_{\widehat{\HHH}}(\bar{\vv}^{(n)} - \alpha_{\bar{\vv}} \nabla f(\bar{\vv}^{(n)})),
\end{align}
where $\PPP_{\widehat{\HHH}}(\y)$ is the operator of projecting $\y$ on $\widehat{\HHH}$. 

Since $f(\vv)$ is not convex, $f(\tilde{\vv}^{(n+1)})$ may not improve the objective sequence, i.e.,  $f(\tilde{\vv}^{(n+1)}) > f(\vv^{(n)})$. However, to speed up convergence, we accept $\vv^{(n+1)} = \tilde{\vv}^{(n+1)}$ if the objective value $f(\tilde{\vv}^{(n+1)})$ is smaller than $c^{(n)}$ which is a relaxation of $f(\vv^{(n)})$ but not far from $f(\vv^{(n)})$. Following \cite [Sect. 3.3]{li15NIPS}, we apply the nonmonotone APG method to find the suitable projection point. Define $c^{(n)}$ the weighted average of $f(\vv^{(n)})$ as
\begin{align}
    \label{c:n+1}
    c^{(n)} = \frac{\sum_{n=1}^{\kappa} \zeta^{(\kappa - n)} f(\vv^{(n)}) }{\sum_{n=1}^{\kappa} \zeta^{(\kappa - n)}},
\end{align}
where $\zeta\in[0,1)$. In each iteration, $c^{(n+1)}$ can be computed
\begin{align}
    \label{b}
    & b^{(n+1)} = \zeta b^{(n)} + 1,
    \\
    \label{c}
    & c^{(n+1)} = \frac {\zeta b^{(n)} c^{(n)} + f(\vv^{(n)}) }{b^{( n + 1)}},
\end{align}
where $c^{(1)} = f(\vv^{(1)})$ and $b^{(1)} = 1$. If $f(\tilde{\vv}^{(n+1)}) \leq c^{(n)} - \zeta\| \tilde{\vv}^{(n+1)} - \bar{\vv}^{(n)} \|^2$ does not hold, additional correction steps are used to prevent this event, where $\|\x\|$ is the Euclidean norm of $\x$. Specifically, another point 
\begin{align}
    \label{vhat}
    \hat{\vv}^{(n+1)} = \PPP_{\widehat{\HHH}}(\vv^{(n)} - \alpha_{\vv} \nabla f(\vv^{(n)})),
\end{align}
is computed with a dedicated step size $\alpha_{\vv}$. Then, we update $\vv^{(n+1)}$ by comparing the objective values at $\tilde{\vv}^{(n+1)}$ and $\hat{\vv}^{(n+1)}$ as 
\begin{align}
    \label{vn+1}
    \vv^{(n+1)} \triangleq
    \begin{cases}
      \tilde{\vv}^{(n+1)}, & \text{if $f(\tilde{\vv}^{(n+1)}) \leq f(\hat{\vv}^{(n+1)})$}\\
      \hat{\vv}^{(n+1)}, & \mbox{otherwise}
    \end{cases}.
\end{align}

Since the feasible set $\widehat{\HHH}$ is bounded, it is true that $\nabla f(\vv)$ is Lipschitz continuous with a constant $L$\footnote{Note that the sum of Lipschitz continuous functions, the product of bounded Lipschitz continuous functions, and the maximum of $0$ and a Lipschitz continuous function are all Lipschitz continuous functions. Since $\nabla f(\vv)$ is a sum of such functions, it is Lipschitz continuous.}, i.e.,
\begin{align}
    \label{Lconstant}
    \|\nabla f(\x) - \nabla f(\y)\| \leq L \| \x - \y \|, \forall \x, \y \in \widehat{\HHH}.
\end{align}
Theoretically, the sufficient conditions for the convergence of the APG approach are $\alpha_{\bar{\vv}} < \frac{1}{L}$ and $\alpha_{\vv} < \frac{1}{L}$. However, finding the value of $L$ is challenging due to the complex nature of the objective function $f(\vv)$. Moreover, these conditions are not necessary for practical implementation. In our numerical results,  $\alpha_{\bar{\vv}}$ and $\alpha_{\vv}$ are kept fixed as sufficiently small values and still offer a convergence for  Algorithm~\ref{algAPG}. 

In Algorithm~\ref{algAPG}, the projection in \eqref{vtilde} and \eqref{vhat} is performed by solving the following problem  
\begin{align}
    \label{P:projection}
    \PPP_{\widehat\HHH}(\vv): \underset{\vv\in\RRR^{2MK\times 1}}{\min}\,\, &
    \| \vv - \rr  \|^2\\
    \mathrm{s.t.} \,\,\,\,\,\,\,\,\,
    \nonumber
    & \eqref{thetamk}, \eqref{sumtheta}, \eqref{znonzero}, \eqref{Khat:2},
\end{align}
for any given vector $\rr = [\rr_1^\text{T},\rr_2^\text{T}]^\text{T}\in\RRR^{2MK\times 1}$, where $\rr_1 \triangleq [\rr_{1,1}^\text{T},\dots,\rr_{1,M}^\text{T}]^\text{T}, \rr_{1,m} \triangleq [r_{1,m1}, \dots, r_{1,mK}]^\text{T}, \rr_2 \triangleq [\rr_{2,1}^\text{T},\dots,\rr_{2,M}^\text{T}]^\text{T}, \rr_{2,m} \triangleq [r_{2,m1}, \dots, r_{2,mK}]^\text{T}$. Problem \eqref{P:projection} can be decomposed into two separate subproblems of optimizing $\THeta_m$ and $\z_m$ for each $m$ as 
\begin{align}
    \label{P:projection:theta}
    \underset{\THeta_m\in\RRR^{MK\times 1}}{\min}\,\, &
    \| \THeta_m - \rr_{1,m}  \|^2\\
    \mathrm{s.t.} \,\,\,\,\,\,\,\,\,
    \nonumber
    & \|\THeta_m\|^2 \leq 1, \THeta_m \geq 0,
    \\
    \label{P:projection:z}
    \underset{\z_m\in\RRR^{MK\times 1}}{\min}\,\, &
    \| \z_m - \rr_{2,m}  \|^2\\
    \mathrm{s.t.} \,\,\,\,\,\,\,\,\,
    \nonumber
    & \|\z_m\|^2 \leq \widehat{K}_m, \z_m \geq 0, \z_m \leq 1,
\end{align}
where the constraints in problems \eqref{P:projection:theta} and \eqref{P:projection:z} follow \eqref{thetamk}, \eqref{sumtheta}, \eqref{znonzero}, \eqref{Khat:2}. The solution to the problem \eqref{P:projection:theta} is the projection of a given point onto the intersection of a Euclidean ball and the positive orthant, which have a closed-form as 
\cite{li15NIPS, farooq21TCOM}
\begin{align}
    \label{theta:solution}
    & \THeta_m = \frac{1}{\max\left(1,\|[\rr_{1,m}]_{0+}\|\right)} [\rr_{1,m}]_{0+},
\end{align}
where $[\x]_{0+} \triangleq [\max(0,x_1), \dots, \max(0,x_K)]^\text{T}, \forall \x\in\RRR^{K \times 1}$.   

Problem \eqref{P:projection:z} is to compute the projection of a given point onto the intersection of two convex sets $S_1=\{\z_m: \|\z_m\|^2 \leq \widehat{K}_m, \z_m \geq 0\}$ and $S_2 = \{\z_m: \z_m \leq 1\}$. On the one hand, the optimal solution to problem \eqref{P:projection:z} can be obtained by the method of alternating projections 
On the other hand, motivated by the approach in \cite{pierro09}, instead of finding directly the optimal solution to problem \eqref{P:projection:z}, it is natural to approximate the solution to problem \eqref{P:projection:z} by the composition of the projection onto $S_1$ and that onto $S_2$. More specifically, the approximated solution to problem \eqref{P:projection:z} is
\begin{align}
    \label{z:solution}
    & \z_m = \Bigg[\frac{\sqrt{\widehat {K}_m}}{\max\bigg(\sqrt{\widehat{K}_m},\|[\rr_{2,m}]_{0+}\|\bigg)} [\rr_{2,m}]_{0+}\Bigg]_{1-},
\end{align}
where $[\x]_{1-} \triangleq [\min(1,x_1), \dots, \min(1,x_K)]^\text{T}, \forall \x\in\RRR^{K \times 1}$. The closed-form expression in \eqref{z:solution} would further reduce the running time of the APG method compared with that using the method of alternating projections. 
Therefore, we use \eqref{z:solution} in the proposed Algorithm~\ref{algAPG}. 

On the other hand, the gradient $\nabla f(\vv)$ can be written as $\nabla f(\vv) = [(\frac{\partial}{\partial \THeta} f(\vv))^\text{T}, (\frac{\partial}{\partial \z} f(\vv))^\text{T} ]^\text{T}$. Here, 
\begin{align}
    \frac{\partial}{\partial \THeta} f(\vv) = \Bigg[\left(\frac{\partial}{\partial \THeta_1} f(\vv)\right)^\text{T}, \dots, \left(\frac{\partial}{\partial \THeta_M} f(\vv)\right)^\text{T}\Bigg] ^\text{T},
\end{align}
\begin{align}
    \frac{\partial}{\partial \z} f(\vv) = \Bigg[ \bigg(\frac{\partial}{\partial \z_1} f(\vv)\bigg)^\text{T}, \dots, \bigg(\frac{\partial}{\partial \z_M} f(\vv)\bigg)^\text{T} \Bigg] ^\text{T},
\end{align}
where $\frac{\partial}{\partial \THeta_m} f(\vv) = [\frac{\partial}{\partial \theta_{m1}} f(\vv), \dots, \frac{\partial}{\partial \theta_{mK}} f(\vv)]^\text{T} $ and $\frac{\partial}{\partial \z_m} f(\vv) = [\frac{\partial}{\partial z_{m1}} f(\vv), \dots, \frac{\partial}{\partial z_{mK}} f(\vv)]^\text{T} $. Thus, the values of $\frac{\partial}{\partial \theta_{mk}} f(\vv)$ and $\frac{\partial}{\partial z_{mk}} f(\vv)$ are computed by
\begin{align} \label{diff:theta:f1}
\begin{split}
    \frac{\partial}{\partial \theta_{mk}} f(\vv) & = \frac{\tau_p - \tau_c}{\tau_c\log 2} \sum_{i\in\K} \Bigg( \frac{\frac{\partial}{\partial \theta_{mk}}(U_i(\vv) + V_i(\vv))}{(U_i(\vv) + V_i(\vv))} \\ 
    & \qquad \quad - \tfrac{\frac{\partial}{\partial \theta_{mk}} V_i(\vv)}{V_i(\vv)} \Bigg) + \chi \frac{\partial}{\partial \theta_{mk}} \widetilde{Q} (\vv),
\end{split}
\end{align}

\begin{align} \label{diff:z:f1}
\begin{split}
    \frac{\partial}{\partial z_{mk}} f(\vv) & = \frac{\tau_p - \tau_c}{\tau_c\log 2} \sum_{i\in\K} \Bigg( \frac{\frac{\partial}{\partial z_{mk}}(U_i(\vv) + V_i(\vv))}{(U_i(\vv) + V_i(\vv))} \\ 
    & \qquad \quad - \tfrac{\frac{\partial}{\partial z_{mk}} V_i(\vv)}{V_i(\vv)} \Bigg) + \chi \frac{\partial}{\partial z_{mk}} \widetilde{Q} (\vv),
\end{split}
\end{align}
where $\widetilde{Q}(\vv) \triangleq \mu_1 Q(\z) + \mu_2 Q_2(\THeta) + \mu_3 Q_3(\z,\THeta) + \mu_4 Q_4(\THeta,\z)$. $\frac{\partial}{\partial \theta_{mk}} U_i(\vv)$ and $\frac{\partial}{\partial \theta_{mk}} V_i(\vv)$ can be expressed as \eqref{diff:Ui} and \eqref{diff:Vi}, shown in the uppermost section of this page. Moreover, from the definitions of $Q_1, Q_2, Q_3, Q_4$ in \eqref{Q1}--\eqref{Q4}, we can derive the expression of $\frac{\partial}{\partial \theta_{mk}}  \widetilde{Q} (\vv)$ and $\frac{\partial}{\partial z_{mk}}  \widetilde{Q} (\vv)$, which can be observed at the middle of this page. We refer to Appendix B for more details. 

\begin{figure*}
\begin{align} \label{diff:Ui}
\frac{\partial}{\partial \theta_{mk}} U_i(\vv) = 
    \begin{cases}
      2 \left( \sum_{m\in\MM} \sqrt{ \rho_d  (N-|\SSS_m|)  \sigma_{mk}^2 }  \theta_{mk} \right) \times 
      \sqrt{\rho_d  (N-|\SSS_m|) \sigma_{mk}^2 }, \ i = k\\
      \quad \quad \quad \quad \quad \quad \quad \quad \quad 0, \quad \quad \quad \quad \quad \quad \quad \quad \quad \quad \quad \quad \quad \quad \quad \quad \quad i \neq k
    \end{cases},
\end{align}
\hrulefill
\end{figure*}

\begin{figure*}
\begin{align} \label{diff:Vi}
\frac{\partial}{\partial \theta_{mk}} V_i(\vv) = 
    \begin{cases}
      2 \rho_d (\beta_{mk} - \delta_{mk} \sigma_{mk}^2) \theta_{mk}, \ \quad \quad \quad i = k\\
      2 \rho_d (\beta_{mi} - \delta_{mi} \sigma_{mi}^2) \theta_{mk}, \quad \quad \quad \quad i \neq k
    \end{cases},
\end{align}
\hrulefill
\end{figure*}

\begin{figure*}
\begin{align} \label{diff:theta:Q}
\nonumber
    \frac{\partial}{\partial \theta_{mk}}  \widetilde{Q} (\vv) & =  \mu_3 4 \max\left(0,\theta_{mk}^2 - z_{mk}^2\right) \theta_{mk} - \mu_2 \sum_{i\in\K} 2\max\! \big(0, \SE_{\text{QoS}} - \SE_i(\THeta)\big) \times \frac{\tau_c - \tau_p}{\tau_c\log 2}\! \Bigg( \frac{\frac{\partial}{\partial \theta_{mk}}(U_i(\vv) + V_i(\vv))}{(U_i(\vv) + V_i(\vv))}  - \\ & \frac{\frac{\partial}{\partial \theta_{mk}} V_i(\vv)}{V_i(\vv)} \Bigg)  +
    \mu_4 z_{mi}^2 \times \sum_{i\in\K}  2 \max\! \Big(0,\! \sum_{i\in\K}\! z_{mi}^2\SE_{i} (\THeta) - C_{\max} \Big)  \! \times \frac{\tau_c - \tau_p}{\tau_c\log 2}\! \Bigg( \frac{\frac{\partial}{\partial \theta_{mk}}(U_i(\vv) + V_i(\vv))}{(U_i(\vv) + V_i(\vv))} - \frac{\frac{\partial}{\partial \theta_{mk}} V_i(\vv)}{V_i(\vv)} \Bigg),
\end{align}
\hrulefill
\end{figure*}

\begin{figure*}
\begin{align} \label{diff:z:Q}
\begin{split}
 \frac{\partial}{\partial z_{mk}}  \widetilde{Q} (\vv) =  \mu_1 (2z_{mk} - 4z_{mk}^3) - \mu_3 4 \max(0,\theta_{mk}^2 - z_{mk}^2) z_{mk}
    & - \mu_3 4\max\! \Big(0,1 - \sum_{m\in\MM} z_{mk}^2\Big) z_{mk}\\
    & + \mu_4 4 \max\! \Big(0, \sum_{i\in\K}\! z_{mi}^2\SE_{i} (\THeta) - C_{\max} \Big) z_{mk} \SE_k(\vv).
\end{split}
\end{align}
\hrulefill
\end{figure*}

In each iteration, the APG-based Algorithm~\ref{algAPG} only requires computing the gradient and projecting a point into the feasible $\widehat\HHH$ with closed-form solutions \cite{farooq21TCOM}. The complexity of $f(\vv)$ is $\mathcal{O}(MK^{2})$; hence, the complexity of $\nabla f(\vv)$ is also $\mathcal{O}(MK^{2})$. Moreover, the complexity of projection operations in \eqref{theta:solution} and \eqref{z:solution} is $\mathcal{O}(MK)$ since $\x\in\RRR^{K \times 1}$ for a given AP. To this end, the complexity in each iteration of the derived APG algorithm is $\mathcal{O}(MK^{2})$. Hence, APG has a significantly lower computational complexity in comparison to SCA.

\section{Performance Evaluation} \label{Sec:simulation}
In this section, we conduct numerical simulations to evaluate the performance of the proposed approaches in small-scale and large-scale CFmMIMO systems with different numbers of APs and UEs, respectively. 

\subsection{Simulation Setup}
We consider a CFmMIMO network, where the APs and UEs are randomly located in a square of $1\times 1$ km$^2$. This square is wrapped-around to emulate a CFmMIMO network with an infinite area. The distances between adjacent APs are at least $50$ m. 
We set the number of antennas in each AP as $N=2$, the coherence block $\tau_c\!=\!200$ samples, and the fronthaul threshold $C_{\max} = 20$ bit/s/Hz, $\forall m$. The QoS SE is $\SE_{\text{QoS}}=0.2$ bit/s/Hz or $4$ Mbps for a bandwidth of $20$ MHz. This is, for example, the requirement for the live-streamed media on appliances of HD video (1080P) \cite{2021_HDVideo}. As stated in\cite{HienTWC2017}, the noise power = bandwidth $\times$ $K_{B}$ $\times$ $T_{0}$ $\times$ noise figure (W), where $K_{B} = 1.381 \times 10^{-23}$ (Joule/Kelvin) is the Boltzmann constant, and $T_{0} = 290$ (Kelvin) is the noise temperature. This work sets the noise figure as $9$ dB and bandwidth as $20$ MHz. Hence, the noise power is $-92$ dBm. The large-scale fading coefficients, $\beta_{mk}$, are modeled similarly to \cite[Eqs. (37), (38)]{emil20TWC}. Let $\tilde{\rho}_d\!=\!1$ W and $\tilde{\rho}_p\!=\!0.1$ W be the maximum transmit power of the APs and uplink pilot symbols, respectively. The maximum transmit powers $\rho_d$ and $\rho_p$ are normalized by the noise power. For the small-sale system, we consider $M \in \{25,36 \}$, $K \in \{5,7 \}$, and $\widehat{K}_m = \widehat{K} \in \{3,5 \}$. We compare the performance under different APs in the large-scale systems, where $M \in \{150, 300 \}$, $K=40$, and $\widehat{K}_m = \widehat{K} = 15$. Each AP $m$ chooses its subset $\SSS_m$ by selecting the UEs that contribute at least $\mu\%$ of the overall channel gain, i.e., $\sum_{k\in\SSS_m} \frac{\beta_{mk}}{\sum_{\ell\in\K} \beta_{m\ell}} \geq \mu \%$ \cite{Giovanni20TWC} while $\mu$ is adjusted to guarantee $N > |\SSS_m|$. In the Algorithm~\ref{algAPG}, we set $\Delta = 2$, $\alpha_{\bar{\vv}} = \alpha_{\vv} = 10^{-4}$, $\zeta = 10^{-1}$, and $\upsilon = 10^{-2}$. 

We use Algorithm 1 with the input of $28,000$ $\B$ to generate $28,000$ data samples $\left( \B, \boldsymbol{\Theta}^{\text{opt}}, \boldsymbol{\Lambda}^{\text{opt}}\right)$ for the training dataset.
We use Algorithm 1 with the input of $3,100$ different values of $\B$ to generate $3,100$ data samples $\left( \B, \boldsymbol{\Theta}^{\text{opt}}, \boldsymbol{\Lambda}^{\text{opt}}\right)$ for the test dataset. 

The batch size $b = 512$, the number of ResNet blocks $R=3$, the learning rate varies between $10^{-3}$ to $10^{-5}$, the maximum training epoch $e_{max}=400$, and the momentum $m_t=0.99$. We train our designed JointCFNet using Algorithm~\ref{algJointCFNet} on an Intel (R) i$7-9800$X CPU with an Nvidia GeForce RTX $2080$ Ti. For channel realizations wherein the obtained solution after post-processing the output of JointCFNet does not satisfy any of constraints \eqref{thetamk}, \eqref{sumtheta}, \eqref{a}-\eqref{atleastoneAP:cons}, the SEs of all UEs in that channel realization are set to zero. To evaluate the effectiveness of our proposed schemes (\textbf{SCA}, \textbf{JointCFNet}, and \textbf{APG}), we compare them with the following methods:
\begin{itemize}
        \item \textbf{Full UE association (FULL)}: All the UEs are served by all the AP, i.e., $a_{mk}=1, \forall m,k$, without constraints  \eqref{fronthaul:cons}--\eqref{atleastoneAP:cons}. We then apply our algorithm to optimize the PC coefficients $\THeta$ under the transmit power constraint for each AP.
	\item \textbf{Heuristic (HEU)}: First, each UE is associated to the AP that has the strongest gains to guarantee             
        \eqref{atleastoneAP:cons} and is different from the associated APs of other UEs. After this step, let $\kappa_m$ be the number of UEs that is associated to AP $m$. To guarantee \eqref{Khat}, each AP $m$ fills up its set of $\widehat{K}$ UEs to serve by selecting $(\widehat{K} - \kappa_m)$ UEs that have the strongest channel gains. PC coefficients $\THeta$ are optimized similarly as \textbf{FULL}. 
\end{itemize}
The UE association approaches discussed in \cite{buzzi2017downlink,nguyen2017energy,mai2018pilot,Gio2019TWC,dan2020} do not consider the maximum number of UEs served by one AP in \eqref{Khat} as well as the maximum fronthaul signaling load, and hence, are not compared with our proposed schemes. 


\subsection{Performance of Small-Scale Systems}
Figure~\ref{Fig:learning} depicts the training and validation losses of the designed \textbf{JointCFNet} for different numbers of APs and UEs. We can observe that both the training and validation loss curves decrease when increasing the number of epochs. Although there is around $0.2$ gap between the training and validation losses, the trained model provides high approximation accuracy compared with the conventional \textbf{SCA} algorithm (see Figs.~\ref{Fig:SumSE_M36}--\ref{Fig:PerSE_M25}). Therefore, we can conclude that the final model fits the training data well without overfitting or underfitting.
\begin{figure}[t!]
 \centering
 {\includegraphics[width=0.49\textwidth]{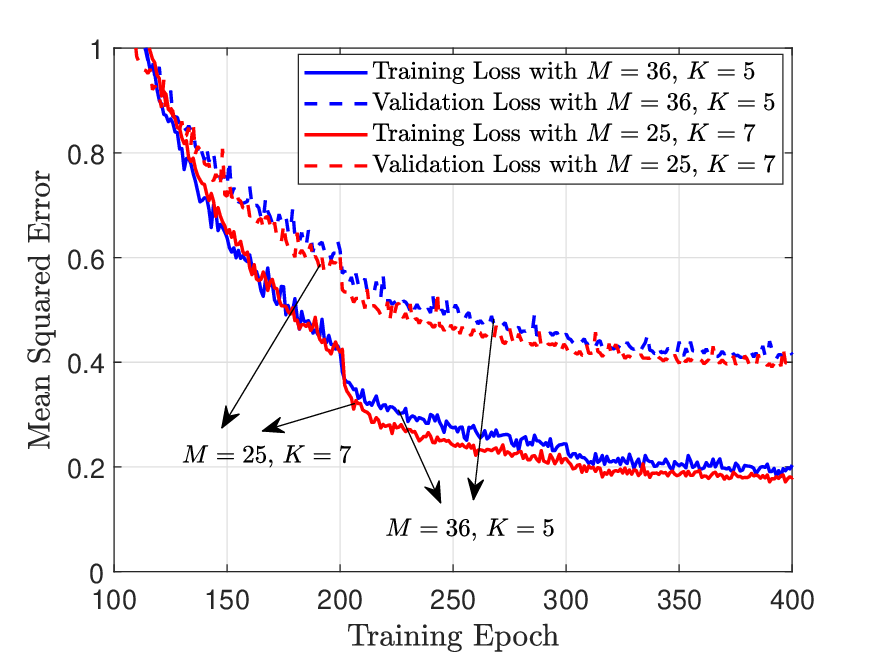}}
 \vspace{-5mm}
 \caption{The training curves of MSE with $M=25, K=7, \widehat{K}=5$, and $M=36, K=5, \widehat{K}=3$.}
 \label{Fig:learning}
 \vspace{-5mm}
\end{figure}
\begin{figure}[t!]
 \centering
 {\includegraphics[width=0.49\textwidth]{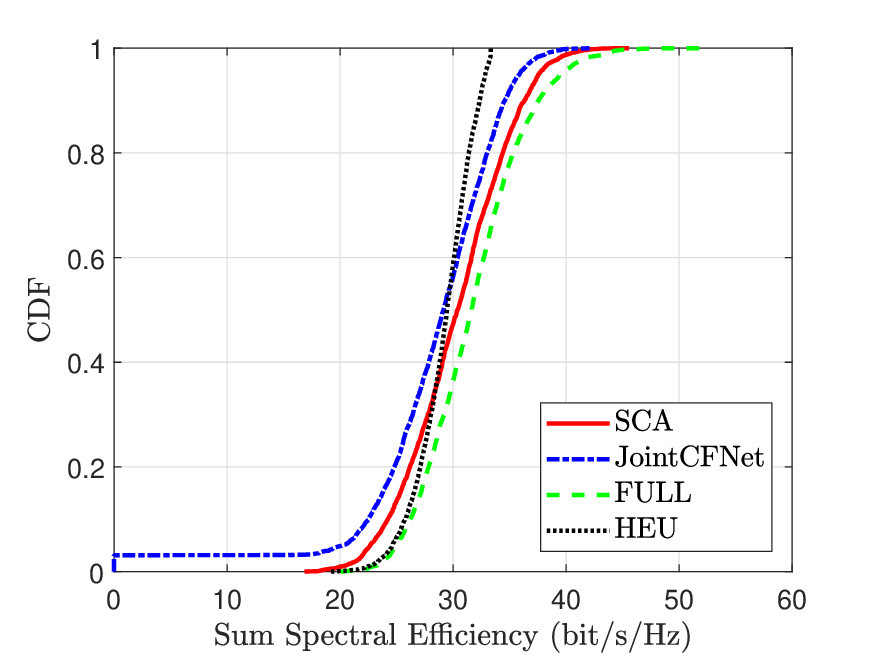}}
 \vspace{-5mm}
 \caption{CDF of the sum SE in a small-scale CFmMIMO system with $M=36, K=5$, and $\widehat{K}=3$.}
 \label{Fig:SumSE_M36}
 \vspace{-0mm}
\end{figure}

Figures~\ref{Fig:SumSE_M36} and~\ref{Fig:SumSE_M25} demonstrate the CDF of the sum SE of the systems with different numbers of APs and UEs. As seen, \textbf{FULL} provides the best performance because all the APs serve all the UEs. \textbf{JointCFNet} provides feasible solutions with high probabilities, i.e., up to $95\%$. This confirms the effectiveness of our \textbf{JointCFNet} solution under the complicated nature of the original optimization problem \eqref{P:SE}.  
Figure~\ref{Fig:SumSE_M36} shows small performance gaps between the proposed schemes compared with \textbf{FULL}. Specifically, the median sum SE of the \textbf{JointCFNet} is up to $97\%$ that of \textbf{SCA}. 
Although \textbf{HEU} and \textbf{JointCFNet} have the same performance, \textbf{JointCFNet} runs much faster than \textbf{HEU}. We refer to Section \ref{Sec:runtime} for the discussion on the run time comparison. 

\begin{figure}[t!]
 \centering
 {\includegraphics[width=0.49\textwidth]{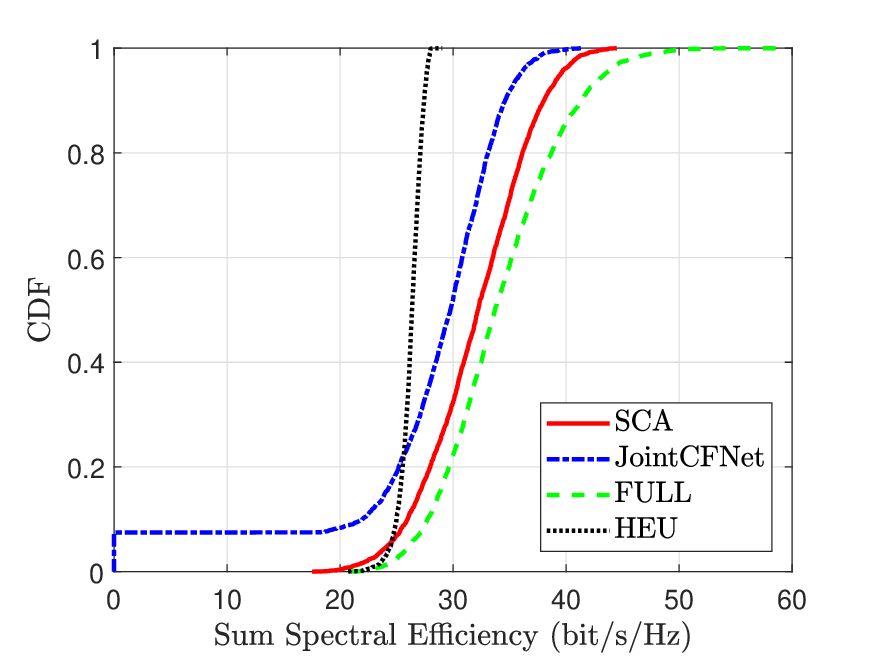}}
 \vspace{-5mm}
 \caption{CDF of the sum SE in a small-scale CFmMIMO system with $M=25, K=7$, and $\widehat{K}=5$.}
 \label{Fig:SumSE_M25}
 \vspace{-0mm}
\end{figure}
\begin{figure}[t!]
 \centering
 {\includegraphics[width=0.49\textwidth]{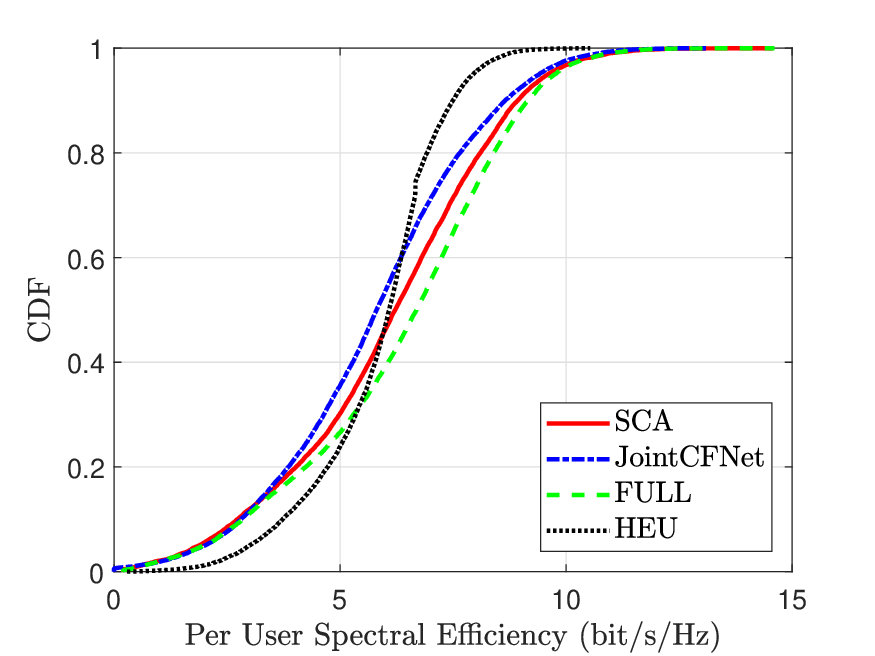}}
 \vspace{-5mm}
 \caption{CDF of the per UE SE in a small-scale CFmMIMO system with $M=36, K=5$, and $\widehat{K}=3$.}
 \label{Fig:PerSE_M36}
 \vspace{-0mm}
\end{figure}

We plot the CDF of the sum SE with fewer APs, a larger number of UEs, and more UEs served by each AP  in Fig.~\ref{Fig:SumSE_M25}. The SE obtained by the \textbf{JointCFNet} is slightly worse than that of the \textbf{SCA} solution. More specifically, in terms of the median value, the sum SE of \textbf{JointCFNet} is around $30$ bit/s/Hz, while the \textbf{SCA} scheme is around $31.5$ bit/s/Hz. To further reduce the gap between the \textbf{JointCFNet} and \textbf{SCA} schemes, we can increase the number of ResNet and Conv blocks to extract more channel features. On the other hand, 
collecting more data can help it learn more about the system's propagation environment. Importantly, the \textbf{HEU} method now presents the worst performance. This is because each AP serves $\widehat{K}$ UEs. Since 
increasing $\widehat{K}$ causes more UE interference, our proposed \textbf{SCA}, \textbf{JointCFNet}, and \textbf{APG} show their significant advantage in managing interference by optimizing PC and UE association over \textbf{HEU}.

The CDF of the per-UE SE is shown in Fig.~\ref{Fig:PerSE_M36} with $M=36$, $K=5$, and $\widehat{K}=3$. We can observe that the \textbf{JointCFNet} model achieves similar performance as the \textbf{SCA} method. 
In terms of median SE, the \textbf{HEU} has a similar performance compared with the \textbf{JointCFNet}. Meanwhile, the per-UE SE of \textbf{JointCFNet} is slightly smaller than  those of \textbf{SCA} and \textbf{FULL} (i.e., $0.5$ bit/s/Hz and $1$ bit/s/Hz). 
Figure~\ref{Fig:PerSE_M25} presents the CDF of the  per UE SE with $M=25$, $K=7$, and $\widehat{K}=5$. 
Now, \textbf{JointCFNet} has a slight performance reduction compared with Fig.~\ref{Fig:PerSE_M36} due to user interferences. Furthermore, with the advantage of managing interference, our proposed \textbf{SCA}, \textbf{JointCFNet}, and \textbf{APG} outperforms \textbf{HEU} under a larger value of $\widehat{K}$. Specifically, at the median point, \textbf{HEU} provides the per-UE SE of $3.6$ bit/s/Hz, while the per-UE SE of \textbf{JointCFNet} and \textbf{SCA}  is $4$ bit/s/Hz, $4.2$ bit/s/Hz, and that of the \textbf{FULL} is $5$ bit/s/Hz.

Note that all the schemes provide nearly identical performance in the small-scale system as shown in Figs.~\ref{Fig:SumSE_M36}-~\ref{Fig:PerSE_M25}. That is reasonable because in a small-scale system, the numbers of APs and UEs are small, and the density of the APs and UEs is small. The distances between UEs and APs are large, and hence, the UE SEs are small. The UA needs to ensure that the least favorable UEs (i.e., those that have large large-scale fading coefficients) achieve the QoS SE while sacrificing the SEs of the most favorable UEs. Therefore, the impact of UA is marginal, making the performances of all the considered schemes in the small-scale system nearly the same.

\begin{figure}[t!]
 \centering
 {\includegraphics[width=0.49\textwidth]{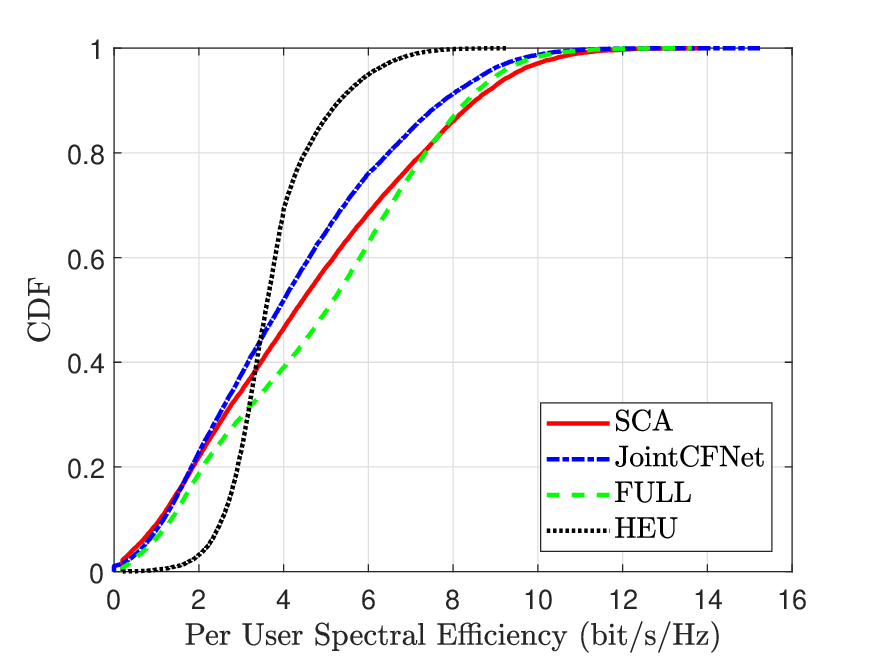}}
 \vspace{-5mm}
 \caption{CDF of the per UE SE in a small-scale CFmMIMO system with $M=25, K=7$, and $\widehat{K}=5$.}
 \label{Fig:PerSE_M25}
 \vspace{-0mm}
\end{figure}

\subsection{Performance of Large-Scale Systems}
Figure~\ref{Fig:Iteration} shows the sum SE versus the number of iterations for \textbf{SCA} (Algorithm~\ref{alg}) and \textbf{APG} (Algorithm~\ref{algAPG}) with different numbers of APs. It demonstrates that both \textbf{SCA} and \textbf{APG} methods converge. We can see that \textbf{SCA}  requires fewer iterations than \textbf{APG}. However, \textbf{SCA} uses the convex solver to obtain a solution that takes a long time during each iteration, while \textbf{APG} is computationally efficient and can be implemented with closed-form expressions. Therefore, \textbf{APG} runs significantly faster than \textbf{SCA}. The run time comparison will be displayed in Section \ref{Sec:runtime}. Moreover, \textbf{APG} offers a sum SE performance that is almost identical to that of \textbf{SCA} under a large number of APs ($M=300$). It shows that \textbf{APG} is suitable for integration in large-scale CFmMIMO systems.
\begin{figure}[t!]
 \centering
 {\includegraphics[width=0.49\textwidth]{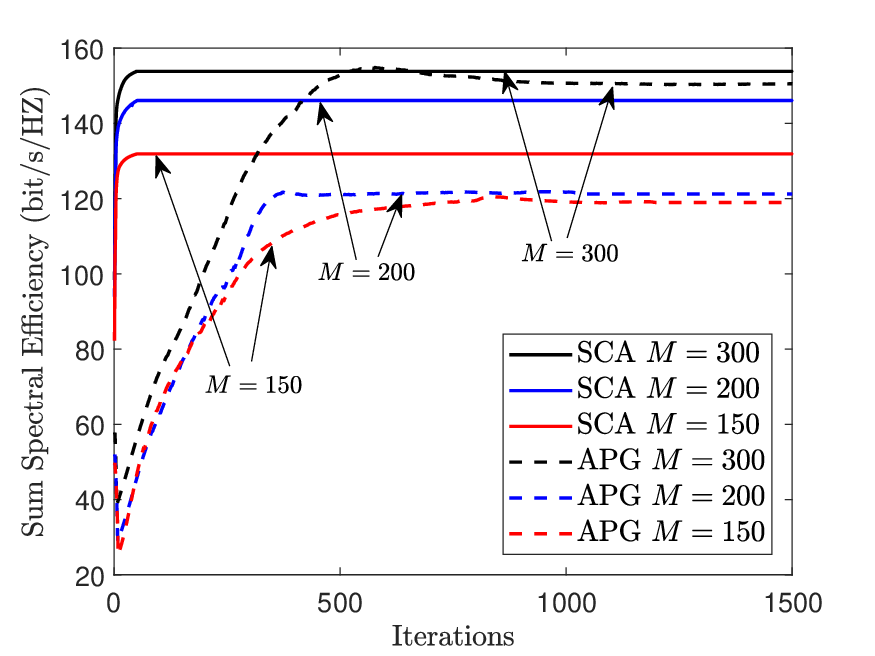}}
 \vspace{-5mm}
 \caption{The convergence of \textbf{SCA} and \textbf{APG} algorithms with different numbers of APs, $M\in \{150,200,300\}$, $K=40$, $\widehat{K}=15$.}
 \label{Fig:Iteration}
 \vspace{-0mm}
\end{figure}

Figures~\ref{Fig:SumSEM150} and~\ref{Fig:SumSEM300} display the CDF of the sum SE with different numbers of APs, when the total number of UEs is $40$, and the maximum number of UEs served by each AP is $15$. In terms of the median sum SE, \textbf{SCA} and \textbf{APG} significantly outperform \textbf{HEU}, while closely approaching \textbf{FULL}. In particular, \textbf{APG} increases the median sum SE by substantial amounts compared with that of \textbf{HEU}, e.g., by up to $219\%$ with $M=150$ and $278\%$ with $M=300$. Moreover, the sum SEs of \textbf{SCA} and \textbf{APG} are close to that of \textbf{FULL}, i.e., up to $82\%$ with $M=300$. These results show the significant advantage of joint optimization of UE association and PC to improve the SE of user-centric CFmMIMO systems. As also seen in Fig.~\ref{Fig:SumSEM150} and  Fig.~\ref{Fig:SumSEM300}, \textbf{APG} can provide a sum SE that is close to that of \textbf{SCA}. The median sum SE of \textbf{APG} can approach up to $95\%$ and $99\%$ that of \textbf{SCA} with $M=150$ and  $M=300$, respectively. Note that for given the same $K=40$ and $\widehat{K}=15$, the performance of \textbf{HEU} is similar at both $M=150$ and  $M=300$. This is because \textbf{HEU} does not manage interference as well as our proposed schemes. Increasing the total number of APs to a value far greater than $K$ and $\widehat{K}$ does not improve its SE performance. 

\begin{figure}[t!]
 \centering
 {\includegraphics[width=0.49\textwidth]{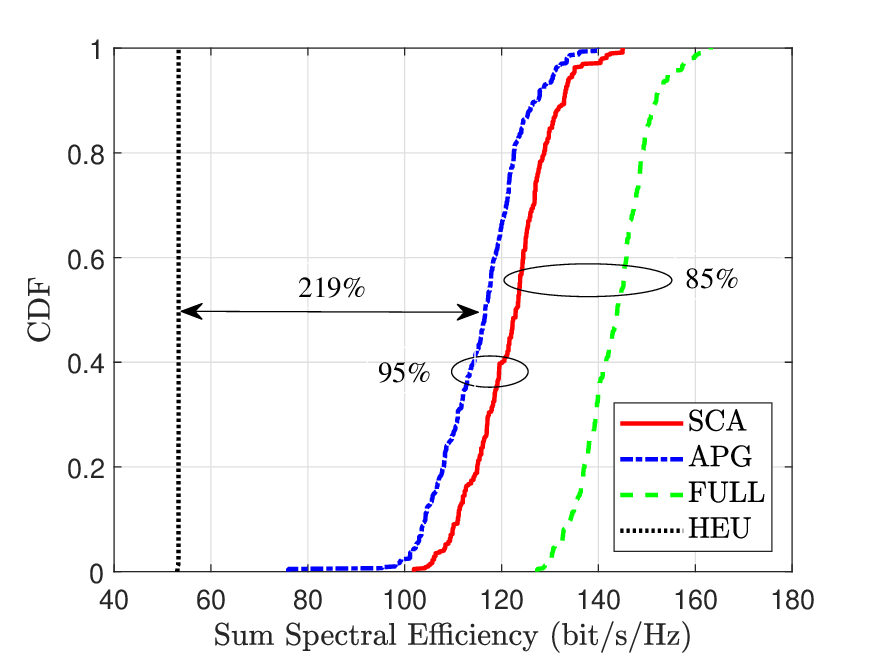}}
 \vspace{-5mm}
 \caption{CDF of the sum SE in a large-scale CFmMIMO system with $M=150$, $K=40$, and $\widehat{K}=15$.}
 \label{Fig:SumSEM150}
 \vspace{-0mm}
\end{figure}

\begin{figure}[t!]
 \centering
 {\includegraphics[width=0.49\textwidth]{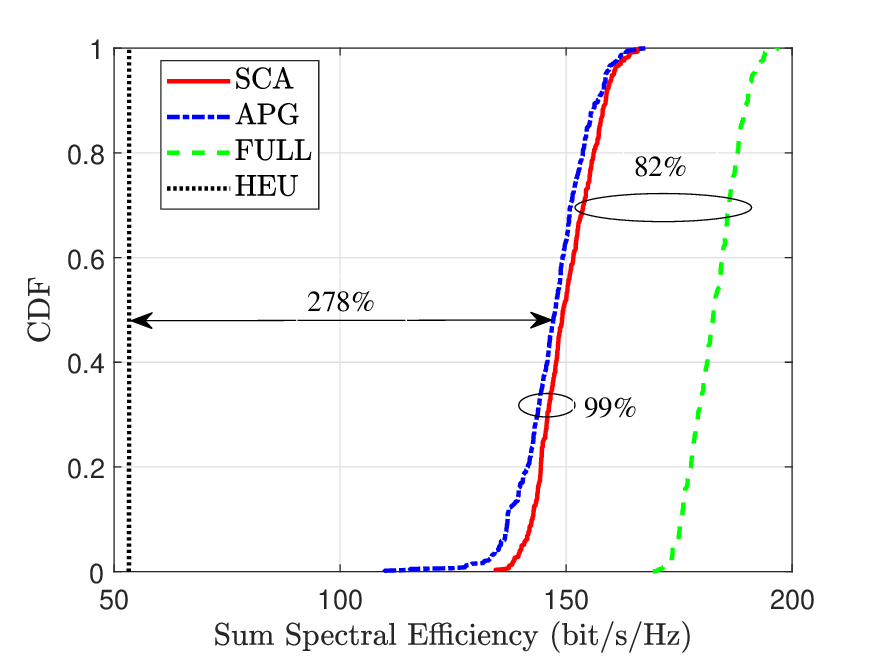}}
 \vspace{-5mm}
 \caption{CDF of the sum SE in a large-scale CFmMIMO system with $M=300$, $K=40$, and $\widehat{K}=15$.}
 \label{Fig:SumSEM300}
 \vspace{-0mm}
\end{figure}

\begin{figure}[t!]
 \centering
 {\includegraphics[width=0.49\textwidth]{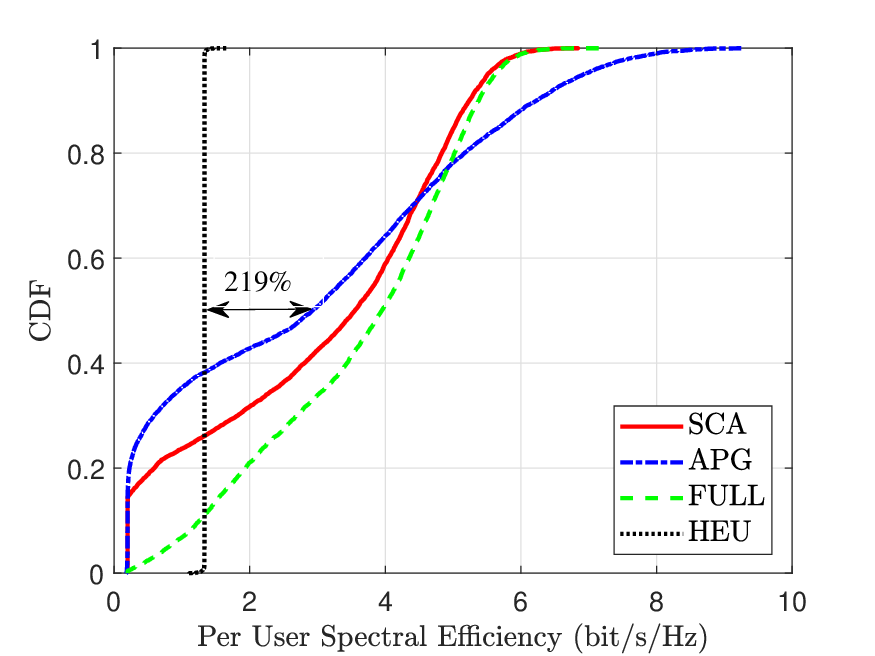}}
 \vspace{-5mm}
 \caption{CDF of the per UE SE in a large-scale CFmMIMO system with $M=150$, $K=40$, and $\widehat{K}=15$.}
 \label{Fig:PerSEM150}
 \vspace{-0mm}
\end{figure}
\begin{figure}[t!]
 \centering
 {\includegraphics[width=0.49\textwidth]{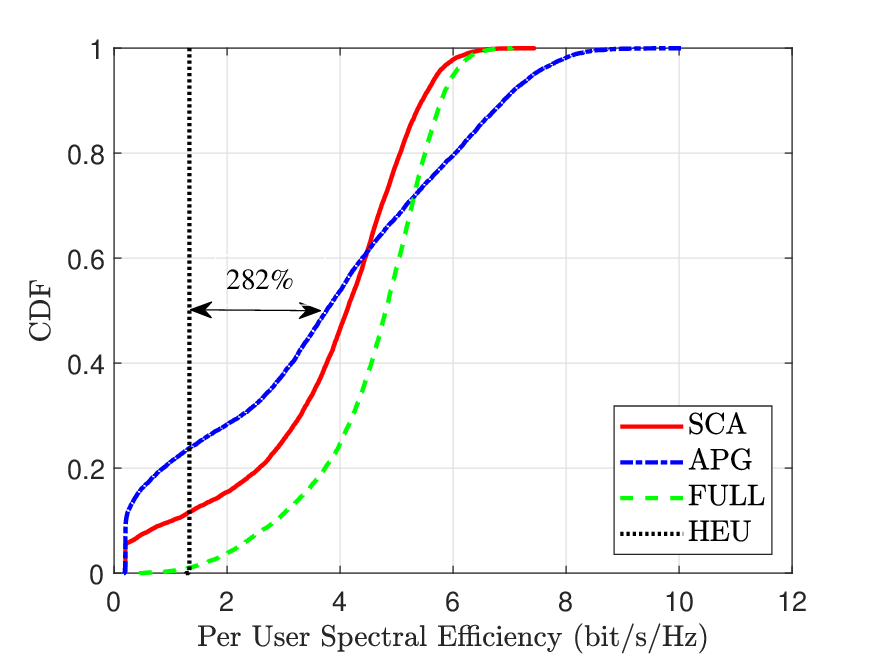}}
 \vspace{-5mm}
 \caption{CDF of the per UE SE in a large-scale CFmMIMO system with $M=300$, $K=40$, and $\widehat{K}=15$.}
 \label{Fig:PerSEM300}
 \vspace{-0mm}
\end{figure}

To gain more insights, Fig.~\ref{Fig:PerSEM150} and Fig.~\ref{Fig:PerSEM300} compare the per-UE SE of \textbf{SCA}, \textbf{APG}, \textbf{FULL}, and \textbf{HEU}.  There is a slight degradation in the median performance of \textbf{APG} compared with \textbf{SCA}, e.g., $3$ bit/s/Hz with $M=150$, and $3.9$ bit/s/Hz with $M=300$, while the \textbf{SCA} method yields $3.7$ bit/s/Hz with $M=150$ and $4.1$ bit/s/Hz with $M=300$. We can observe that the performance of \textbf{APG} is very close to the \textbf{SCA} under large numbers of APs (the performance gap is less than $0.2$ bit/s/Hz with $M=300$). Moreover, \textbf{APG} offers a median per-UE SEs significantly higher than that of \textbf{HEU}, e.g., up to $219\%$ with $M=150$ and $282\%$ with $M=300$. Similarly, increasing $M$ has no effect on the performance of \textbf{HEU} when $M$ is already very large.

Note that the performances of these schemes are significantly different in the large-scale system as shown in Figs.~\ref{Fig:SumSEM150}\mbox{-}\ref{Fig:PerSEM300}. This is because in a large-scale system, the numbers of APs and UEs are large, and the density of the APs and UEs is high. The distances between UEs and APs are small, and the UE SEs are naturally large and almost larger than the QoS SEs. The UA now focuses on improving the SE of the most favorable UEs. Therefore, the impact of UA is large, which makes the performance of all the considered schemes in the large-scale system significantly different.

\begin{table*}[!t]
\renewcommand{\arraystretch}{1.0}
\caption{
Average run time (seconds) for small-scale CFmMIMO systems
}
\label{Tab:small-scale}
\begin{center}
\vspace{-4mm}
\begin{tabular}{|c|c|c|c|c|c|}
\hline
System specifications & \textbf{SCA} & \textbf{FULL} & \textbf{HEU} & \textbf{JointCFNet} & Average run time ratio of \textbf{SCA} over \textbf{JointCFNet}\\
\hline
$M=25$, $K=7$, $\widehat{K}=5$ & $25.727$ & 16.197 & 11.252 & \textbf{0.002} & 12850\\
\hline
$M=36$, $K=5$, $\widehat{K}=3$ & 19.222 & 16.138 & 9.278  & \textbf{0.002} & 9611\\
\hline
\end{tabular}
\end{center}
\vspace{-2mm}
\end{table*}

\begin{table*}[!t]
\renewcommand{\arraystretch}{1.0}
\caption{
Average run time (seconds) for large-scale CFmMIMO systems
}
\label{Tab:large-scale}
\begin{center}
\vspace{-4mm}
\begin{tabular}{|c|c|c|c|c|c|}
\hline
System specifications & \textbf{SCA} & \textbf{FULL} & \textbf{HEU} & \textbf{APG} & Average run time ratio of \textbf{SCA} over \textbf{APG}\\
\hline
$M=150$, $K=40$, $\widehat{K}=15$ & 151.735 & 44.813 & 57.025 & \textbf{14.921} & 10\\
\hline
$M=300$, $K=40$, $\widehat{K}=15$ & 256.171 & 67.912 & 73.414  & \textbf{21.919} & 12\\
\hline
\end{tabular}
\end{center}
\vspace{-2mm}
\end{table*}

\subsection{Run Time} \label{Sec:runtime}
We compare the run time of the considered methods to evaluate the computational complexity. All algorithms are implemented with a single CPU core on the Intel (R) i$7-9800$X CPU. Table~\ref{Tab:small-scale} tabulates the average run time in a small-scale CFmMIMO system with $3000$ channel realizations. The run time of steps \eqref{theta:star} and \eqref{SEcheck} for post-processing the output obtained from the trained model is much smaller than the run time of the proposed \textbf{JointCFNet} in the online phase, and hence, is ignored in Table~\ref{Tab:small-scale}. 
For a system with $M=25$, $K=7$, and $\widehat{K}=5$, the run time is around $25.7$s using the \textbf{SCA} algorithm, while \textbf{FULL} requires $16.2$s and the \textbf{HEU} method takes $11.3$s to obtain a solution. However, the run time of the proposed \textbf{JointCFNet} is only $2$ms, around $10^{4}$ times faster than the \textbf{SCA} method. For a system with $M=36$, $K=5$, and $\widehat{K}=3$, the run time of \textbf{SCA}, \textbf{FULL}, \textbf{HEU}, and \textbf{JointCFNet} is $19.2$s, $16.1$s, $9.3$s, and $2$ms, respectively. Since the \textbf{JointCFNet} model only performs the forward propagation (addition and multiplication operations) under the trained parameters. Thus, the computational complexity is significantly reduced compared with the other considered schemes.

The average run time of a large-scale system with 400 channel realizations is given in Table~\ref{Tab:large-scale}. In a large-scale system with $M=150$, $K=40$, and $\widehat{K}=15$, the run time of \textbf{SCA} is $151.7$s, while \textbf{FULL} and \textbf{HEU} are around $44.8$s and $57$s, respectively. The proposed \textbf{APG} only requires $14.9$s to complete the joint optimization task. The average run time of \textbf{APG} is $10$ fold smaller than that of \textbf{SCA}. If we increase the number of APs, all algorithms require a longer optimization time. However, we can observe that the run time of \textbf{APG} is still $12$ times shorter than that of \textbf{SCA}. This means \textbf{APG} is significantly faster than \textbf{SCA} in large-scale systems.

\vspace{-4pt}
\section{Conclusion} \label{Sec:conclusion}
In this work, we proposed a joint optimization approach of UA and PC for CFmMIMO systems with local PPZF precoding. We formulated a mixed-integer nonconvex optimization problem to maximize the sum SE under requirements on per-AP transmit power, QoS rate, fronthaul capacity, and the maximum number of UEs each AP serves. By utilizing SCA and DL techniques, we proposed two novel schemes to solve the formulated problem in small-scale CFmMIMO systems. Then, we presented a low-complexity APG method to obtain a suboptimal solution for large-scale systems. Numerical results showed that in a small-scale CFmMIMO system, the DNN-based JointCFNet can achieve comparable performance to the SCA method while significantly reducing computational complexity. In a large-scale CFmMIMO system, the presented APG algorithm can significantly increase the SE compared with the heuristic approaches and obtain nearly the same SE as the SCA method with considerably lower complexity. These findings highlight the practicality of selecting optimization algorithms based on the system scale. Specifically, for a small-scale system, the SCA and JointCFNet can provide similar performance, but the latter is much faster than the former. On the other hand, the APG approach can provide acceptable performance with reduced run time in large-scale systems. Finally, we point out that the joint PC and UA in the uplink transmission is of importance and a timely research topic for future research.


%

\appendices
\section{}
Using the first-order Taylor series expansion at a given point ${a_{mk}^{(n)}}$, the convex upper bound of \eqref{Q} is given by
\begin{align} \label{Taylor:Q1}
\begin{split}
    \widehat{Q} (\aaa) & \triangleq \sum_{k\in\K} \sum_{m\in\MM} \Big[ a_{mk} - \left((a_{mk}^{(n)})^2 + 2a_{mk}^{(n)} (a_{mk} - a_{mk}^{(n)})\right) \Big] \\ 
    & = \sum_{k\in\K} \sum_{m\in\MM} \Big( a_{mk} - 2a_{mk}^{(n)} a_{mk} + (a_{mk}^{(n)})^2 \Big).
\end{split}
\end{align}
Similarly, $V_k(\THeta)$ in \eqref{Vhat} at the given point $\theta_{m\ell}^{(n)}$ can be expressed as
\begin{align} \label{Taylor:V1}
\begin{split}
    V_k(\THeta) &\triangleq  \sum_{\ell\in\K} \sum_{m\in\MM} \rho_d \Big[ (\theta_{m\ell}^{(n)})^2  +   2\theta_{m\ell}^{(n)} (\theta_{m\ell} -\theta_{m\ell}^{(n)})\Big] \nu + 1\\
    & = \sum_{\ell\in\K} \sum_{m\in\MM} \rho_d \Big[2\theta_{m\ell}^{(n)} \theta_{m\ell} - (\theta_{m\ell}^{(n)})^2 \Big] \nu + 1,
\end{split}
\end{align}
where $ \nu \triangleq ( \beta_{mk} - \delta_{mk} \sigma_{mk}^2)$. Plugging \eqref{Taylor:V1} into \eqref{Vhat} yields \eqref{Vhat:convex}.

Let $f({x,y}) \triangleq {(x-y)}^2$ 
be a convex function in $(x,y)$. By invoking the first-order Taylor series expansion at a given point $(x^{(k)}, y^{(k)})$, we have the convex lower bound of $f({x,y})$ as
\begin{align} \label{TaylorLowerBound:Derive:1}
   2(x^{(k)} - y^{(k)})(x-y) - {(x^{(k)} - y^{(k)})^2} \leq {(x-y)^2},
\end{align} 
which together with $4xy = \left(x+y\right)^{2} - \left(x-y\right)^{2}$ yields  
\begin{align} \label{TaylorLowerBound:Derive:3}
\begin{split}
    4xy \leq {(x+y)^2} - 2(x^{(k)} - y^{(k)})(x-y)  + {(x^{(k)} - y^{(k)})^2}.
\end{split}
\end{align}
Lastly, substituting \eqref{TaylorLowerBound:Derive:3} into the first-order Taylor series expansion at the point $({a_{mk}^{(n)},t_{k}^{(n)}})$  of \eqref{fronthaul:cons:2} gives \eqref{QoS:cons:2:convex}.

\section{}
The value of $\nabla f(\vv)$ is given as
\begin{align} \label{diff:f}
\begin{split}
\nabla f(\vv) & \!=\! \Bigg[\left(\frac{\partial}{\partial \THeta} f(\vv)\right)^\text{T}, \left(\frac{\partial}{\partial \z} f(\vv)\right)^\text{T}\Bigg]^\text{T} \\
& \!=\! \Bigg[ \left(\frac{\partial}{\partial \THeta_{1}} f(\vv)\right)^\text{T},\dots, \left(\frac{\partial}{\partial \z_{M}} f(\vv)\right)^\text{T} \Bigg]^\text{T},
\end{split}
\end{align}
where
\begin{align} \label{diff:theta:m}
\frac{\partial}{\partial \THeta_{m}} f(\vv) = \Bigg[ \left(\frac{\partial}{\partial \theta_{m1}} f(\vv)\right),\dots, \left(\frac{\partial}{\partial \theta_{mK}} f(\vv)\right) \Bigg]^\text{T},
\end{align}
\begin{align} \label{diff:z:m}
\frac{\partial}{\partial \z_{m}} f(\vv) = \Bigg[ \left(\frac{\partial}{\partial z_{m1}} f(\vv)\right),\dots, \left(\frac{\partial}{\partial z_{mK}} f(\vv)\right) \Bigg]^\text{T}.
\end{align}
Substituting $f(\vv)$ into \eqref{diff:theta:m} and \eqref{diff:z:m} we can obtain
\begin{align} \label{diff:theta:f}
    \frac{\partial}{\partial \theta_{mk}} f(\vv) = - \sum_{i\in\K}  \frac{\partial}{\partial \theta_{mk}} \SE_i (\vv) + \chi \frac{\partial}{\partial \theta_{mk}} \widetilde{Q} (\vv),
\end{align}

\begin{align} \label{diff:z:f}
    \frac{\partial}{\partial z_{mk}} f(\vv) = - \sum_{i\in\K}  \frac{\partial}{\partial z_{mk}} \SE_i (\vv) + \chi \frac{\partial}{\partial z_{mk}}  \widetilde{Q} (\vv).
\end{align}
From \eqref{DownlinkSE}, $\SE_i (\vv)$, $\forall i\in\K$, can be rewritten as 
\begin{align}
    \!\!\!\SE_i (\vv) 
    = \!\frac{\tau_c-\tau_p}{\tau_c} \Big[ \log_2 \left(U_i(\vv) + V_i(\vv)\right) - \log_2 V_i(\vv) \Big].
\end{align}
Then, we have
\begin{align} \label{diff:theta:SE}
    & \!\!\!\!\!\frac{\partial}{\partial \theta_{mk}}\! \SE_i (\vv) \!=\! \frac{\tau_c \!-\! \tau_p}{\tau_c\log 2}\! \Bigg[\! \frac{\frac{\partial}{\partial \theta_{mk}}\left(U_i(\vv) + V_i(\vv)\right)}{U_i(\vv) + V_i(\vv)} \!-\! \frac{\frac{\partial}{\partial \theta_{mk}} V_i(\vv)}{V_i(\vv)} \Bigg], 
\end{align}

\begin{align} \label{diff:z:SE}
    & \!\!\!\!\! \frac{\partial}{\partial z_{mk}}\! \SE_i (\vv) \!=\! \frac{\tau_c \!-\! \tau_p}{\tau_c\log 2}\! \Bigg[\! \frac{\frac{\partial}{\partial z_{mk}} \left(U_i(\vv) + V_i(\vv)\right)}{U_i(\vv) + V_i(\vv)} \!-\! \frac{\frac{\partial}{\partial z_{mk}} V_i(\vv)}{V_i(\vv)} \!\Bigg].
\end{align}
Note that $\frac{\partial}{\partial z_{mk}} U_i(\vv) = 0, \frac{\partial}{\partial z_{mk}} V_i(\vv) = 0, \forall m,k,i$. Plugging \eqref{diff:theta:SE} and \eqref{diff:z:SE} into \eqref{diff:theta:f} and \eqref{diff:z:f} gives \eqref{diff:theta:f1} and \eqref{diff:z:f1}.

Further, we can decompose $\frac{\partial}{\partial \theta_{mk}} U_i(\vv)$ as two parts:
\subsubsection{$i = k$}
\begin{align} \label{diff:theta:u2}
\nonumber
\frac{\partial}{\partial \theta_{mk}} U_i(\vv) & = \frac{\partial}{\partial \theta_{mk}} \Bigg( \sum_{m\in\MM} \sqrt{ \rho_d  (N-|\SSS_m|)  \sigma_{mk}^2 }  \theta_{mk}  \Bigg)^{2} 
\\
\nonumber
& = 2 \Bigg( \sum_{m\in\MM} \sqrt{ \rho_d  (N-|\SSS_m|)  \sigma_{mk}^2 }  \theta_{mk}  \Bigg) \times 
\\
\nonumber
& \quad \quad \quad \frac{\partial}{\partial \theta_{mk}} \Bigg( \sum_{m\in\MM} \sqrt{ \rho_d  (N-|\SSS_m|)  \sigma_{mk}^2 }  \theta_{mk}  \Bigg) 
\\
\nonumber
& = 2 \Bigg( \sum_{m\in\MM} \sqrt{ \rho_d  (N-|\SSS_m|)  \sigma_{mk}^2 }  \theta_{mk}  \Bigg) \times 
\\
& \quad \quad \quad \sqrt{ \rho_d  (N-|\SSS_m|)  \sigma_{mk}^2},
\end{align}
\subsubsection{$i \neq k$}
\begin{align} \label{diff:theta:u3}
\frac{\partial}{\partial \theta_{mk}} U_i(\vv) & = 0. 
\end{align}
Combing \eqref{diff:theta:u2} and \eqref{diff:theta:u3}, we obtain \eqref{diff:Ui}. The procedure to obtain \eqref{diff:Vi}--\eqref{diff:z:Q} is done in the same manner.

\bibliographystyle{IEEEtran}
\bibliography{IEEEabrv}
\end{document}